\documentclass[a4paper,12pt]{article}
\pdfoutput=1
\usepackage{test,caption,subcaption,tikz}
\usetikzlibrary{decorations.pathreplacing,calligraphy}
\usepackage[colorlinks,citecolor=blue,urlcolor=magenta]{hyperref}

\usepackage[style=ext-authoryear-comp,
sorting=nyt,
dashed=false, 
maxcitenames=2, 
maxbibnames=99, 
uniquelist=false,
uniquename=false,
giveninits=true, 
natbib, 
date=year 
]{biblatex}

\AtBeginRefsection{\GenRefcontextData{sorting=ynt}}
\AtEveryCite{\localrefcontext[sorting=ynt]}

\DeclareFieldFormat{pages}{#1} 
\renewbibmacro{in:}{\ifentrytype{article}{}{\printtext{\bibstring{in}\intitlepunct}}} 

\DeclareFieldFormat[article,inbook,incollection,inproceedings,patent,thesis,unpublished]{titlecase:title}{\MakeSentenceCase*{#1}} 

\usepackage[inline,shortlabels]{enumitem}
\setlist[enumerate,1]{label=(\roman*)}

\usepackage[margin = 1.25in]{geometry}
\usepackage{setspace}
\title{Housing Bubbles with Phase Transitions\thanks{We thank seminar participants at J\"onk\"oping, Laval, McGill, Oslo, Princeton, Rochester, Royal Holloway, Tokyo, UCSD, Waseda, and various conferences for valuable comments and feedback.}}
\author{Tomohiro Hirano\thanks{Department of Economics, Royal Holloway, University of London, and Research Associate at the Center for Macroeconomics at the London School of Economics,  \href{mailto:tomohiro.hirano@rhul.ac.uk}{tomohiro.hirano@rhul.ac.uk}.} \and Alexis Akira Toda\thanks{Department of Economics, Emory University, \href{mailto:alexis.akira.toda@emory.edu}{alexis.akira.toda@emory.edu}.}}

\numberwithin{equation}{section}
\numberwithin{lem}{section}
\numberwithin{prop}{section}

\newcommand{\cS}{\mathcal{S}}

\begin{document}

\maketitle

\begin{abstract}
We analyze how equilibrium housing prices are determined in the process of economic development within an overlapping generations model with perfect housing and rental markets. We characterize the rent growth rate in all equilibria. The economy exhibits a two-stage phase transition: as incomes of home buyers rise, the equilibrium regime changes from fundamental to bubble possibility, where fundamental and bubbly equilibria coexist. With even higher incomes, fundamental equilibria disappear and housing bubbles become a necessity. We also discuss extensions and refinements such as equilibrium uniqueness, multiple savings vehicles, welfare implications, credit- and expectation-driven bubbles, and testable implications of our theory.

\medskip


\medskip

\end{abstract}

\section{Introduction}

Over the last three decades, many countries have experienced appreciation in housing prices, with upward trends in the price-rent ratio.\footnote{\label{fn:OECD}See, for instance, Figure 1 of \citet{AmaralDohmenKohlSchularick2024} for 27 major agglomerations in 15 OECD countries and U.S. Metropolitan Statistical Areas. Figure 1 of \citet{Backer-PeralHazellMianYield}, who exploit a natural experiment from long-term lease renewal in U.K., shows a downward trend in the housing yield.} The situation is often referred to in the popular press as a housing bubble. Because fluctuations in housing prices have often been associated with macroeconomic problems, many academics and policymakers want to understand why and how housing bubbles emerge in the first place. However, the mechanism of the emergence of housing bubbles is poorly understood. In addition, theoretically, it is well known that there is a fundamental difficulty in generating asset price bubbles (existence of speculation) in dividend-paying assets such as housing, land, and stocks (see \S\ref{subsec:lit} Related literature for details). The theory of rational bubbles attached to real assets remains largely underdeveloped: at present, there is no theoretical framework for considering whether housing prices reflect fundamentals or contain bubbles.

The primary purpose of this paper is to fill this gap and to present a theory of rational housing bubbles. We are interested in the following questions.
\begin{enumerate*}
    \item What is the mechanism by which equilibrium housing prices \emph{can} or \emph{must} be disconnected from fundamentals in the long term, exhibiting a speculative bubble in a dynamic general equilibrium setting in which housing rents and prices are both endogenously determined?
    \item How is the disconnection related to economic conditions such as the income or access to credit of home buyers and to the formation of expectations about future economic conditions, namely the process of economic development?
    \item What are the welfare properties of equilibria?
\end{enumerate*}

To capture how equilibrium housing prices are determined in the process of economic development, we develop a two-period overlapping generations model with perfect housing and rental markets. The economy is inhabited by overlapping generations that live for two periods (young and old age) and consume two commodities (consumption good and housing service). The ownership and occupancy of a housing unit are separated, so there is a price for house ownership as a financial asset (housing price) and a price for house occupancy as a commodity (rent). All markets are competitive and frictionless. A rational expectations equilibrium consists of a sequence of prices (housing price and rent) and allocations (consumption good, housing stock, and housing service) such that all agents optimize and markets clear. An equilibrium is \emph{fundamental} (\emph{bubbly}) if the housing price equals (exceeds) the present value of rents. In this model, the dividend of housing, namely rent, is endogenous. If housing supply is inelastic, as the economy grows and agents get richer, they increase the demand for housing, which pushes up both the housing price and rent. Under these circumstances, it is not obvious whether housing prices will grow faster than rents and a housing bubble emerge: the possibility or necessity of housing bubbles becomes a nontrivial question.

We obtain three main results.  First, we identify the theoretical mechanism of generating housing bubbles, which crucially depends on the income of home buyers and the elasticity of substitution between consumption and housing. We prove that the economy experiences a \emph{two-stage phase transition} in the process of economic development, which is captured by the long-run income ratio of the young (home buyers) relative to the old (home sellers). When the income ratio is sufficiently low, housing bubbles cannot arise and a fundamental equilibrium exists, which we refer to as the \emph{fundamental regime}. When the income ratio rises and exceeds the first critical value, a phase transition occurs.\footnote{\emph{Phase transition} is a technical term in natural sciences that refers to a discontinuous change in the state as we continuously change a parameter. For instance, as we increase the temperature, the matter (\eg, $\mathrm{H_2O}$) changes from solid (ice) to liquid (water) to gas (vapor). The analogy here is appropriate because the regime of the economy abruptly changes from fundamental to bubbly as income rises. \citet{ColeKehoe2000} show in the context of debt crisis that as the confidence parameter changes, the equilibrium regime changes from a no-crisis zone to a crisis zone and then to a default only zone, producing phase transitions (though they do not use this term).} Both a fundamental and a bubbly equilibrium exist, and the equilibrium is selected by agents' self-fulfilling expectations. We refer to this coexistence region as the \emph{bubble possibility regime}. When the income ratio exceeds the second and still higher critical value, another phase transition takes place to the \emph{bubble necessity regime}, where fundamental equilibria do not exist and housing bubbles become inevitable. Furthermore, we prove the uniqueness of equilibrium under weak conditions. We show that the fundamental equilibrium is always unique, and the bubbly equilibrium is unique if the elasticity of intertemporal substitution is not too much below $1/2$.

The intuition for this two-stage phase transition is the following. Let $G>1$ be the long-run growth rate of the economy and $\gamma>0$ the reciprocal of the elasticity of substitution between consumption and housing, which in the model also equals the elasticity of rent with respect to income. Empirical estimates suggest $\gamma<1$,\footnote{\label{fn:elasticity}\citet[Table 2]{OgakiReinhart1998} estimate the elasticity of substitution between durable and nondurable goods using aggregate data and obtain $\gamma=1/1.24=0.81$. \citet[Appendix C]{PiazzesiSchneiderTuzel2007} estimate a cointegrating equation between the price and quantity of housing service relative to consumption using aggregate data and obtain $\gamma=1/1.27=0.79$. \citet[Table 2]{HowardLiebersohn2021} estimate $\gamma=0.79$ using cross-sectional data on income and rents.} and a theoretical argument also supports it: if $\gamma>1$, as the economy grows and agents get richer, the young asymptotically spend all income on housing, the price-rent ratio converges to zero, and the interest rate diverges to infinity, which are all pathological and counterfactual. Since $\gamma=1$ (Cobb-Douglas) is a knife-edge case, it is natural to focus on the case $\gamma<1$. Under this condition, by equating marginal utility to prices, consumption grows at rate $G$ but the rent grows at rate $G^\gamma<G$. Therefore, if the housing price only reflects fundamentals in the long-run equilibrium, it must also grow at rate $G^\gamma$. Since housing price grows slower than endowments in any fundamental equilibrium, the expenditure share of housing converges to zero in the long run and the interest rate $R$ is pinned down as the marginal rate of intertemporal substitution in the autarky allocation. If $R>G^\gamma$, a fundamental equilibrium exists. If $R<G^\gamma$, a fundamental equilibrium cannot exist, for otherwise the fundamental value of housing (the present value of rents) becomes infinite, which is obviously impossible in equilibrium. Therefore, as the young become richer and the interest rate falls below a certain threshold, the fundamental equilibrium becomes unsustainable, and a housing bubble \emph{inevitably} emerges. Fundamental and bubbly equilibria coexist when the autarkic interest rate satisfies $G^\gamma<R<G$, which corresponds to an intermediate range for the income ratio of the young.

As our second main result, using the two-stage phase transition and uniqueness of equilibrium dynamics, we present expectation-driven housing booms containing a bubble and their collapse. In our model, because agents are forward-looking and housing prices reflect information about future economic conditions, whether bubbles arise or not in equilibrium depends on long-run expectations about the income ratio of home buyers. As long as agents expect economic development and high incomes in the future, housing prices start rising now and contain a bubble, even if the current income of home buyers is low and the economy appears to stay in the fundamental region. During this dynamics driven by optimistic beliefs, the price-income ratio and the price-rent ratio simultaneously rise, and hence the housing price dynamics may appear unsustainable because prices grow faster than incomes. On the other hand, if these optimistic expectations do not materialize, the bubble collapses. This expectation-driven housing bubble occurs as the unique equilibrium outcome.

Our third main result concerns the existence of dynamic inefficiency in an economy with a productive non-reproducible asset and welfare analysis of housing bubbles. Take the famous \citet{Diamond1965} model, which considers an economy without a productive non-reproducible asset like land. This model shows that under some conditions, dynamically inefficient equilibria can arise. However, \citet{McCallum1987} shows that the introduction of land eliminates dynamically inefficient equilibria, thereby resolving the concerns (over-savings problem) raised by \citet{Diamond1965} and restoring Pareto efficiency. Since this result, it has been widely believed that in OLG models with land, dynamically inefficient equilibria would not arise \citep{Mountford2004}. We theoretically show that this commonly-accepted understanding is not necessarily true: dynamically inefficient equilibria can robustly arise even with housing, which plays the role of a productive non-reproducible asset. Moreover, the existence of dynamic inefficiency has a \emph{non-monotonic} relationship to the income ratio of the young (home buyers) relative to the old (home sellers). If the income ratio is high or low enough (corresponding to the bubble necessity and fundamental regimes, respectively), the economy exhibits dynamic efficiency. Dynamically inefficient equilibria arise only in the intermediate range of the income ratio (corresponding to the bubble possibility regime).

We emphasize that we obtain these results and draw new insights from what could be called the simplest possible model of housing. We thus see our paper as a fundamental theoretical contribution that could be used as a stepping stone for constructing more realistic models aimed for empirical or quantitative analysis.

\subsection{Related literature}\label{subsec:lit}

Our paper is related to the literature on the valuation of housing. Unlike quantitative models reviewed in \citet{PiazzesiSchneider2016}, our primary interest is to study conditions under which housing \emph{can} or \emph{must} be overvalued, in the sense that equilibrium housing prices contain a speculative aspect. Our paper belongs to the so-called ``rational bubble literature'' that studies bubbles as speculation, which was pioneered by \citet{Samuelson1958}, \citet{Bewley1980}, \citet{Tirole1985}, \citet{ScheinkmanWeiss1986}, \citet{Kocherlakota1992}, and \citet{SantosWoodford1997}. Theoretical foundations and applications of rational bubble include \citet{HuangWerner2000}, \citet{CaballeroKrishnamurthy2006}, 
\citet{BloiseCitanna2019}, and \citet*{BrunnermeierMerckelSannikov2024}, among others.\footnote{See \citet{HiranoToda2024JME} for a recent review of the rational bubble literature. \citet{BrunnermeierOehmke2013} survey the broader literature with alternative approaches including heterogeneous beliefs \citep{ScheinkmanXiong2003,FostelGeanakoplos2012Tranching}, asymmetric information \citep{AbreuBrunnermeier2003,Barlevy2014,AllenBarlevyGale2022}, liquidity \citep{LagosRocheteauWright2017,BranchPetrosky-NadeauRocheteau2016}, among others. See \citet{HiranoToda2025EJW} for a discussion of other approaches as well as the confusion in the literature.}

It is well known in the rational bubble literature that there is a fundamental difficulty in generating bubbles in dividend-paying assets: there is a \emph{discontinuity} in proving the existence of a bubble between zero-dividend assets (pure bubble assets like fiat money or cryptocurrency) and dividend-paying assets (like housing). This difficulty follows from \citet[Theorem 3.3,  Corollary 3.4]{SantosWoodford1997}, who show that, when the asset pays nonnegligible dividends relative to the aggregate endowment, bubbles are impossible. This ``Bubble Impossibility Theorem'' has been extended under alternative financial constraints by \citet{Kocherlakota2008} and \citet{Werner2014}. Due to the fundamental difficulty, the rational bubble literature has almost exclusively focused on pure bubbles without dividends. While pure bubble models are useful in describing money or cryptocurrency, as \citet[\S4.7]{HiranoToda2024JME} argue, pure bubble models are subject to criticisms such as
\begin{enumerate*}
    \item the lack of realism due to zero dividends,
    \item the lack of robustness due to equilibrium indeterminacy (\ie, the existence of a continuum of pure bubble equilibria), and
    \item the inability to connect to the large empirical literature that uses dividends to test whether asset prices reflect fundamentals \citep{Shiller1981, PhillipsShi2020}.
\end{enumerate*}
Our model circumvents all these issues because housing endogenously generates positive rents in a way consistent with \citet{SantosWoodford1997}. We note that pure bubble models often consider a (hypothetical) economy in which there is no housing or land in the first place and then show that under some conditions (\eg, a rise in young's income), pure bubbles can arise (possibility), and equilibria are indeterminate. However, bubbles inevitably emerge (necessity) if and only if we consider an economy with housing from the beginning, and equilibria are determinate. Therefore, the economic insights are markedly different between the cases with zero and positive rents.

Within the rational bubble literature, there are several papers that study housing bubbles, including \citet{CaballeroKrishnamurthy2006}, \citet{Kocherlakota2009,Kocherlakota2013}, \citet{ArceLopez-Salido2011}, \citet{Zhao2015}, \citet{ChenWen2017}, and \citet{GraczykPhan2021}. However, in these papers, either housing does not generate housing services or the rental market is missing and housing do not generate rents, so the fundamental value of housing is zero, which is essentially the same as pure bubbles. Furthermore, most of these papers employ logarithmic utility, which corresponds to the case $\gamma=1$ in our model. As we show in Appendix \ref{subsec:gamma=1}, under this common but knife-edge parameter specification, housing bubbles do not arise if housing generates rents. Hence there is another discontinuity in generating housing bubbles between the cases with zero and positive rents.

Due to the aforementioned fundamental difficulty of attaching bubbles to dividend-paying assets, there are only a handful papers that treat this topic. \citet[\S7]{Wilson1981} provides the first example of bubbles attached to dividend-paying assets (see \citet[Example 1]{HiranoToda2025JPE} for more discussion of this example). \citet[Proposition 1(c)]{Tirole1985} recognizes that, with dividend-paying assets, bubbles could be necessary for equilibrium existence if the bubbleless interest rate is less than the dividend growth rate. There are important differences from \citet{Tirole1985} and our results.
\begin{enumerate*}
    \item First, Tirole introduces a dividend-paying asset into the \citet{Diamond1965} OLG model and assumes exogenous and constant dividends. In contrast, in our model housing prices and rents are both endogenous. Under this circumstance, it is not obvious whether housing prices \emph{can} or \emph{must} grow faster than rents, \ie, whether housing bubbles \emph{can} or \emph{must} arise.
    \item Second, and more importantly, the recent paper of \citet{PhamTodaTirole} revisits Tirole's model and shows that his results require some qualifications. In particular, they provide a counterexample to \citet[Proposition 1(c)]{Tirole1985} based on a closed-form solution, in which the unique equilibrium is fundamental.
\end{enumerate*} 

\citet{HiranoToda2025JPE} establish the concept of the necessity of bubbles in modern macro-finance models including OLG models and Bewley-type infinite-horizon models. \citet{HiranoToda2025PNAS} prove a bubble necessity theorem in economies with aggregate risk when land is used as a production factor and the productivities and the elasticity of substitution in the production function satisfy some conditions. Our results build on these earlier papers (see, for instance, Lemma \ref{lem:neccesity}) but there are important differences.
\begin{enumerate*}
    \item Although dividends are exogenous in \citet{HiranoToda2025JPE}, in our model rents are endogenous. As noted in the introduction, this difference is significant. Nevertheless, we characterize the long-run rent growth rate in \emph{all} equilibria (Theorem \ref{thm:Gr}), and we identify the importance of the income ratio between the young and old and the elasticity of substitution between consumption and housing for endogenously satisfying the bubble necessity condition.
    \item \citet{HiranoToda2025JPE} do not study how the equilibria look like, but we provide a complete analysis of the long-run behavior of equilibria using the local stable manifold theorem.
    \item Unlike \citet{HiranoToda2025PNAS}, who focus on the role of the supply side (productivities and elasticity of substitution, both associated with the production function) for the emergence of land price bubbles, we focus on the role of the demand side for housing in generating housing bubbles, namely the income of home buyers and the elasticity of substitution between consumption and housing.\footnote{In addition, \citet{HiranoToda2025PNAS} assume Cobb-Douglas utility with the old having zero income for analytical tractability, whereas the preferences and endowments in our model are general.}
    \item We derive a new insight that the determination of housing prices changes significantly at different stages of economic development. Using this insight, in \S\ref{subsec:expectation}, we analyze the role of expectations in the formation and bursting of housing bubbles, and then in \S\ref{sec:conclude}, we discuss testable implications that empirical researchers can exploit to test our theory.
\end{enumerate*}

Finally, from a theoretical perspective, the dynamics involving housing in the process of economic development in our model is characterized by unbalanced growth, which is closely related to the dynamics in the literature on structural transformation \citep{Baumol1967,Matsuyama1992,AcemogluGuerrieri2008,BueraKaboski2012,FujiwaraMatsuyama2024}. In this literature, there are two approaches to the factors that generate unbalanced growth: one focusing on the demand side and the other on the supply side \citep[\S21.1-2]{Acemoglu2009}. In our model, unbalanced growth occurs due to demand factors for housing. A critical difference from the literature is that we derive asset pricing implications under unbalanced growth, whereas the literature abstracts away from asset prices. To our knowledge, the present paper would be the first to simultaneously show unbalanced growth and the emergence of bubbles attached to real assets due to demand factors.

\section{Model}\label{sec:model}

\subsection{Primitives}\label{subsec:model_primitive}

Time is discrete and indexed by $t=0,1,\dotsc$. We consider a deterministic overlapping generations (OLG) economy in which agents live for two periods (young and old age) and demand a consumption good and housing service. We employ an OLG model because it allows us to capture life-cycle behaviors regarding housing demand in a simple setting.

\paragraph{Commodities, asset, and endowments}

There are two perishable commodities (consumption good and housing service) and a durable non-reproducible asset (housing stock) in the economy. The housing service is the right to occupy a housing unit between two periods. Every period, one unit of housing stock inelastically produces one unit of housing service. The time $t$ endowment of the consumption good is $e_t^y>0$ for the young and $e_t^o>0$ for the old. At $t=0$, the housing stock (whose aggregate supply is normalized to 1) is owned by the old.

\paragraph{Preferences}

An agent born at time $t$ lives for two periods and has utility function $U(c_t^y,c_{t+1}^o,h_t)$, where $c_t^y>0$ is consumption when young, $c_{t+1}^o>0$ is consumption when old, and $h_t>0$ is housing service consumed when transitioning from young to old. As usual, we assume that $U:\R_{++}^3\to \R$ is continuously differentiable, has strictly positive first partial derivatives, is strictly quasi-concave, and satisfies Inada conditions to guarantee interior solutions. The initial old care only about their consumption $c_0^o$.

\paragraph{Markets}

We consider an ideal world in which the ownership and occupancy of housing are separated and traded at competitive frictionless markets: agents trade housing (a financial asset) only to store value (transfer resources across time), whereas they purchase housing service (a commodity) only to derive utility.\footnote{Therefore, nothing prevents agents from purchasing a mansion as an investment while renting a campsite to sleep, or vice versa. Owner-occupants can be thought of agents who rent the houses they own to themselves. However, because in our model agents within a generation are homogeneous, in equilibrium each young agent demands one unit of housing and one unit of housing service, so the agents end up being owner-occupants.}

Let $r_t$ be the price of housing service (rent) and $P_t$ be the housing price (excluding current rent) quoted in units of time $t$ consumption. Let $x_t$ denote the demand for the housing stock. Then the budget constraints of generation $t$ are
\begin{subequations}\label{eq:budget}
\begin{align}
    &\text{Young:} & c_t^y+P_tx_t+r_th_t&\le e_t^y, \label{eq:budget_young}\\
    &\text{Old:} & c_{t+1}^o&\le e_{t+1}^o+(P_{t+1}+r_{t+1})x_t.\label{eq:budget_old}
\end{align}
\end{subequations}
The budget constraint of the young \eqref{eq:budget_young} states that the young spend income on consumption, purchase of housing stock, and rent. The budget constraint of the old \eqref{eq:budget_old} states that the old consume the endowment and the income from renting and selling housing.

\paragraph{Equilibrium}

As usual, an equilibrium is defined by individual optimization and market clearing.

\begin{defn}\label{defn:eq}
A \emph{rational expectations equilibrium} consists of a sequence of prices $\set{(P_t,r_t)}_{t=0}^\infty$ and allocations $\set{(c_t^y,c_t^o,h_t,x_t)}_{t=0}^\infty$ such that for each $t$,
\begin{enumerate*}
    \item (Individual optimization) the young maximize utility $U(c_t^y,c_{t+1}^o,h_t)$ subject to the budget constraints \eqref{eq:budget},
    \item (Commodity market clearing) $c_t^y+c_t^o=e_t^y+e_t^o$,
    \item (Rental market clearing) $h_t=1$,
    \item (Housing market clearing) $x_t=1$.
\end{enumerate*}
\end{defn}

Note that because the old exit the economy, the young are the natural buyers of housing, which explains the housing market clearing condition $x_t=1$.

\subsection{Equilibrium and housing bubble}\label{subsec:model_eqcond}

We characterize the equilibrium and define housing bubbles. Using the rental and housing market clearing conditions $h_t=x_t=1$ and the budget constraint \eqref{eq:budget}, we obtain
\begin{equation}
    (c_t^y,c_t^o)=(e_t^y-P_t-r_t,e_t^o+P_t+r_t)=(e_t^y-S_t,e_t^o+S_t), \label{eq:yz}
\end{equation}
where $S_t\coloneqq P_t+r_t$ is total expenditure on housing. Throughout the paper, we refer to $P_t$ as the \emph{housing price} and $S_t$ as the \emph{housing expenditure}. Let
\begin{equation}
R_t\coloneqq \frac{P_{t+1}+r_{t+1}}{P_t}=\frac{S_{t+1}}{P_t} \label{eq:R}
\end{equation}
be the implied gross risk-free rate between time $t$ and $t+1$. Then the two budget constraints in \eqref{eq:budget} can be combined into one as
\begin{equation}
    c_t^y+\frac{c_{t+1}^o}{R_t}+r_th_t\le e_t^y+\frac{e_{t+1}^o}{R_t}.\label{eq:budget_combined}
\end{equation}
In what follows, to simplify notation, we often use $(y,z)$ in place of $(c^y,c^o)$ and hence write $U(y,z,h)$ instead of $U(c^y,c^o,h)$.\footnote{The mnemonic is that $y$ is the first letter of ``young'' and $z$ is the next alphabet.} Letting $\lambda_t\ge 0$ be the Lagrange multiplier associated with the combined budget constraint \eqref{eq:budget_combined}, we obtain the first-order conditions
\begin{equation}
    (U_y,U_z,U_h)=\lambda_t(1,1/R_t,r_t), \label{eq:foc}
\end{equation}
where we use the shorthand for partial derivatives $U_y\coloneqq \partial U/\partial y$, $U_z\coloneqq \partial U/\partial z$, and $U_h\coloneqq \partial U/\partial h$, which are evaluated at
\begin{equation}
    (y,z,h)=(e_t^y-S_t,e_{t+1}^o+S_{t+1},1).\label{eq:demand}
\end{equation}
Using \eqref{eq:foc}, we obtain
$1/R_t=U_z/U_y$ and $r_t=U_h/U_y$. Combining these two equations, the definition of $R_t$ in \eqref{eq:R}, and $S_t=P_t+r_t$, we obtain
\begin{equation}
    S_{t+1}U_z=S_tU_y-U_h,\label{eq:s_dynamics}
\end{equation}
where the partial derivatives of $U$ are evaluated at \eqref{eq:demand}. The following theorem establishes the existence of equilibrium and characterizes equilibrium quantities.

\begin{thm}[Existence and characterization of equilibrium]\label{thm:eq}
The following statements are true.
\begin{enumerate}
\item A rational expectations equilibrium exists.
\item An equilibrium has a one-to-one correspondence with the sequence $\set{S_t}_{t=0}^\infty$ satisfying $0<S_t<e_t^y$ and the nonlinear difference equation \eqref{eq:s_dynamics}.
\item The equilibrium quantities are given by
\begin{subequations}\label{eq:eqobj}
\begin{align}
    (c_t^y,c_t^o)&=(e_t^y-S_t,e_t^o+S_t), \label{eq:eq_yz}\\
    P_t&=S_t-(U_h/U_y)(e_t^y-S_t,e_{t+1}^o+S_{t+1},1), \label{eq:eq_p}\\
     r_t&=(U_h/U_y)(e_t^y-S_t,e_{t+1}^o+S_{t+1},1), \label{eq:eq_r}\\
     R_t&=(U_y/U_z)(e_t^y-S_t,e_{t+1}^o+S_{t+1},1). \label{eq:eq_R}
\end{align}
\end{subequations}
\end{enumerate}
\end{thm}

\begin{proof}
The existence of equilibrium is standard \citep{BalaskoShell1980} and follows from the same argument as the proof of \citet[Theorem 1]{HiranoToda2025JPE}. The equilibrium quantities \eqref{eq:eqobj} follow from the preceding argument.
\end{proof}

By Theorem \ref{thm:eq}, an equilibrium is fully characterized by the sequence of housing expenditures $\set{S_t}_{t=0}^\infty$. For this reason, we often refer to $\set{S_t}_{t=0}^\infty$ as an equilibrium without specifying each object in Definition \ref{defn:eq}.

Following the standard definition of rational bubbles in the literature, we say there is a housing bubble if the housing price exceeds its fundamental value defined by the present value of rents. (See Appendix \ref{sec:bubble} for a self-contained exposition.) Let $R_t>0$ be the equilibrium gross risk-free rate. Let $q_t>0$ be the Arrow-Debreu price of date-$t$ consumption in units of date-0 consumption, so $q_0=1$ and $q_t=1/\prod_{s=0}^{t-1}R_s$. The \emph{fundamental value} of housing is the present value of rents
\begin{equation}
    V_t\coloneqq \frac{1}{q_t}\sum_{s=t+1}^\infty q_sr_s. \label{eq:Vt}
\end{equation}

\begin{defn}
A rational expectations equilibrium is \emph{fundamental} if $P_t=V_t$ for all $t$ and \emph{bubbly} if $P_t>V_t$ for all $t$.
\end{defn}

Appendix \ref{sec:bubble} shows that an equilibrium is either fundamental or bubbly.

\subsection{Additional assumptions}\label{subsec:model_asmp}

To make qualitative predictions, we put more structure by specializing the utility function and endowments.

\begin{asmp}[Endowments]\label{asmp:G}
There exist $G>1$, $e_1,e_2>0$, and $T>0$ such that the endowments are $(e_t^y,e_t^o)=(e_1G^t,e_2G^t)$ for $t\ge T$.
\end{asmp}

Assumption \ref{asmp:G} implies that in the long run, the economy exogenously grows at rate $G>1$ and the income ratio between the young and old is constant. We assume exogenous growth of endowments and fixed supply of housing as the simplest benchmark to illustrate the key mechanism of housing bubbles.\footnote{In Figure \ref{fig:units_GDP} of Appendix \ref{sec:fact}, we document empirical evidence that economic growth is faster than the growth of housing supply. We can extend our model to include endogenous growth, as studied in \citet*{HiranoJinnaiTodaLeverage}, and variable housing supply by introducing the construction of new housing.}

\begin{asmp}[Utility]\label{asmp:U}
The utility function takes the form
\begin{equation}
    U(y,z,h)=u(c(y,z))+mu(h), \label{eq:utility}
\end{equation}
where
\begin{enumerate*}
    \item\label{item:U_c} the composite consumption $c(y,z)$ is homogeneous of degree 1 and quasi-concave,
    \item\label{item:U_u} the utility of composite consumption/housing service is $u(c)=\frac{c^{1-\gamma}}{1-\gamma}$ for some $\gamma>0$ ($u(c)=\log c$ if $\gamma=1$), and
    \item\label{item:U_phi} $m>0$ is a marginal utility parameter.
\end{enumerate*}
\end{asmp}

Assumption \ref{asmp:U}\ref{item:U_c} implies that agents (apart from the initial old) care about consumption $(c^y,c^o)$ only through the homothetic composite consumption $c(c^y,c^o)$, which (together with Assumption \ref{asmp:G}) allows us to study asymptotically balanced growth paths. Assumption \ref{asmp:U}\ref{item:U_u} implies that agents have constant elasticity of substitution $1/\gamma>0$ between consumption and housing service.\footnote{The functional form \eqref{eq:utility} implies that we first aggregate young and old consumption, and then housing service. Our main results are not affected if we change the utility function to $U(y,z,h)=c(u^{-1}(u(y)+mu(h)),z)$ so that we first aggregate young consumption and housing service, and then old consumption.}

Throughout the main text, we focus on the case $\gamma<1$ (so the elasticity of substitution between consumption and housing $1/\gamma$ exceeds 1) and defer the analysis of the case $\gamma\ge 1$ to Appendix \ref{sec:gamma>=1}. There are three reasons for doing so. First, $\gamma<1$ is the empirically relevant case (Footnote \ref{fn:elasticity}). Second, $\gamma=1$ is a knife-edge case. Third, as we show in Proposition \ref{prop:gamma>1}, the equilibrium with $\gamma>1$ is pathological and counterfactual: the young asymptotically spend all income on housing (purchase and rent); the price-rent ratio converges to zero; and the gross risk-free rate diverges to infinity. Hence the case $\gamma>1$ is economically irrelevant.

Since by Assumption \ref{asmp:U}\ref{item:U_c} $c$ is homogeneous of degree 1 and quasi-concave, Theorem 11.14 of \citet[p.~158]{TodaEME} implies that $c$ is actually concave. Because we wish to study smooth interior solutions, we further strengthen the assumption on utility as follows.

\begin{asmp}[Composite consumption]\label{asmp:c}
The composite consumption $c:\R_{++}^2\to (0,\infty)$ is homogeneous of degree 1, twice continuously differentiable, and satisfies $c_y>0$, $c_z>0$, $c_{yy}<0$, $c_{zz}<0$, $c_y(0,z)=\infty$, $c_z(y,0)=\infty$.
\end{asmp}

A typical functional form for $c$ satisfying Assumption \ref{asmp:c} is the constant elasticity of substitution (CES) specification
\begin{equation}
    c(y,z)=\begin{cases*}
        ((1-\beta)y^{1-\sigma}+\beta z^{1-\sigma})^\frac{1}{1-\sigma} & if $0<\sigma\neq 1$,\\
        y^{1-\beta}z^\beta & if $\sigma=1$,
        \end{cases*} \label{eq:CES}
\end{equation}
where $1/\sigma$ is the elasticity of intertemporal substitution and $\beta\in (0,1)$ dictates time preference.

\section{Housing prices in the long run}\label{sec:longrun}

In this section, we study the long-run behavior of equilibrium housing prices.

\subsection{Long-run properties of equilibria}

We present two results that are crucial for the subsequent analysis.

\begin{lem}[Backward induction]\label{lem:backward}
    Suppose Assumptions \ref{asmp:U} and \ref{asmp:c} hold. If $\cS_T=\set{S_t}_{t=T}^\infty$ is an equilibrium starting at $t=T$, there exists a unique equilibrium $\cS_0=\set{S_t}_{t=0}^\infty$ starting at $t=0$ that agrees with $\cS_T$ for $t\ge T$.
\end{lem}

Lemma \ref{lem:backward} shows that once we establish the existence of equilibrium starting at $t=T$, we may uniquely extend the equilibrium path backward in time, which allows us to focus on the long-run behavior of the economy and guarantees the uniqueness of the transitional dynamics. Since by Assumption \ref{asmp:G} the endowments eventually grow at a constant rate $G$, unless otherwise stated, without loss of generality we assume that endowments are $(e_t^y,e_t^o)=(e_1G^t,e_2G^t)$ for all $t$.

In our model, housing rents are endogenous, unlike the setting in \citet{HiranoToda2025JPE}. To apply the Bubble Necessity Theorem (Lemma \ref{lem:neccesity}), we need to characterize the long-run rent growth rate. The following theorem, which is the main technical contribution of this paper, exactly achieves this.

\begin{thm}[Long-run rent growth]\label{thm:Gr}
Suppose Assumptions \ref{asmp:G}--\ref{asmp:c} hold and $\gamma<1$. Then in any equilibrium, the long-run rent growth rate is
\begin{equation}
    G_r\coloneqq \limsup_{t\to\infty}r_t^{1/t}=G^\gamma. \label{eq:Gr}
\end{equation}
\end{thm}

The intuition for Theorem \ref{thm:Gr} is the following. Since endowments grow at rate $G$ and the elasticity of substitution between consumption and housing service is $1/\gamma$, the marginal rate of substitution (which equals rent) must grow at rate $G^\gamma$. Of course, the proof is not straightforward because Theorem \ref{thm:Gr} refers to \emph{any} equilibrium.\footnote{For readers that are curious but do not wish to read the proof, we provide a brief explanation. We first prove an intermediate result $\liminf G^{-t}S_t<e_1$, implying that the young's saving rate is strictly positive. To prove this by contradiction, assume $\liminf G^{-t}S_t\ge e_1$ and hence $G^{-t}S_t\to e_1$ (because $G^{-t}S_t\le e_1$ necessarily by the budget constraint). Then we can show that the equilibrium is \emph{simultaneously} bubbly and fundamental, which is impossible. The rest of the proof is straightforward.}

We next define the long-run equilibrium. By Assumption \ref{asmp:U}, the equilibrium dynamics \eqref{eq:s_dynamics} becomes
\begin{equation}
    S_{t+1}c_z=S_tc_y-m c^\gamma, \label{eq:s_dynamics2}
\end{equation}
where $c,c_y,c_z$ are evaluated at $(y,z)=(e_t^y-S_t,e_{t+1}^o+S_{t+1})$. To study asymptotically balanced growth paths, let $s_t\coloneqq S_t/e_t^y=S_t/(e_1G^t)$ be the housing expenditure normalized by the income of the young. Since $c$ is homogeneous of degree 1, its partial derivatives $c_y,c_z$ are homogeneous of degree 0. Therefore, dividing both sides of \eqref{eq:s_dynamics2} by $e_1G^t$, we obtain
\begin{equation}
    Gs_{t+1}c_z=s_tc_y-me_1^{\gamma-1}G^{(\gamma-1)t}c^\gamma, \label{eq:s_dynamics3}
\end{equation}
where $c,c_y,c_z$ are evaluated at $(y,z)=(1-s_t,G(w+s_{t+1}))$ for the old to young income ratio $w\coloneqq e_2/e_1$.

When $\gamma<1$, the difference equation \eqref{eq:s_dynamics3} explicitly depends on time $t$ (is non-autonomous), which is inconvenient for analysis. To convert it to an autonomous system, define the auxiliary variable $\xi_t=(\xi_{1t},\xi_{2t})$ by $\xi_{1t}=s_t=S_t/(e_1G^t)$ and $\xi_{2t}=e_1^{\gamma-1}G^{(\gamma-1)t}$. Then the one-dimensional non-autonomous nonlinear difference equation \eqref{eq:s_dynamics3} reduces to the two-dimensional autonomous nonlinear difference equation $\Phi(\xi_t,\xi_{t+1})=0$, where
\begin{subequations}\label{eq:H}
    \begin{align}
        \Phi_1(\xi,\eta)&=G\eta_1c_z-\xi_1 c_y+mc^\gamma \xi_2,\\
        \Phi_2(\xi,\eta)&=\eta_2-G^{\gamma-1}\xi_2
    \end{align}
\end{subequations}
and $c,c_y,c_z$ are evaluated at $(y,z)=(1-\xi_1,G(w+\eta_1))$ with $w\coloneqq e_2/e_1$. We can now define a long-run equilibrium.

\begin{defn}\label{defn:lreq}
A rational expectations equilibrium $\set{S_t}_{t=0}^\infty$ is a \emph{long-run equilibrium} if the sequence of auxiliary variables $\set{\xi_t}_{t=0}^\infty$ is convergent.
\end{defn}

If $\xi_t\to \xi$, since $G>1$ and $\gamma\in (0,1)$, we have $\Phi(\xi,\xi)=0$ if and only if $\xi_2=0$ and $\xi_1(Gc_z-c_y)=0$, where $c_y,c_z$ are evaluated at $(y,z)=(1-\xi_1,G(w+\xi_1))$. Clearly $\xi_f^*\coloneqq (0,0)$ is a steady state of $\Phi$, which we refer to as the \emph{fundamental} steady state.\footnote{The terminology ``steady state'' is subtle. The steady state $\xi^*$ corresponds to the detrended system $\Phi(\xi_t,\xi_{t+1})=0$, not the original economy. The long-run equilibrium in the original economy is an asymptotically balanced growth path that corresponds to a particular sequence $\set{\xi_t}_{t=0}^\infty\subset \R_{++}^2$ converging to $\xi^*$ that is consistent with the initial conditions and the equilibrium condition $\Phi(\xi_t,\xi_{t+1})=0$.} In order for $\Phi$ to have a nontrivial ($\xi_1=s>0$) steady state, which we refer to as the \emph{bubbly} steady state, it is necessary and sufficient that $Gc_z-c_y=0$.

\subsection{(Non)existence of fundamental equilibria}

As a benchmark, we start our analysis with the existence, and possibly nonexistence, of fundamental equilibria. By Theorem \ref{thm:Gr}, the rent must asymptotically grow at rate $G^\gamma$. Hence if the housing price equals its fundamental value (present value of rents), it must also grow at rate $G^\gamma$. But since endowments grow faster at rate $G>G^\gamma$, the expenditure share of housing converges to zero in the long run and the consumption allocation becomes autarkic: $(c_t^y,c_t^o)\sim (e_1G^t,e_2G^t)$. This argument suggests that in any fundamental equilibrium, the interest rate behaves like
\begin{equation}
    R_t=\frac{c_y}{c_z}(c_t^y,c_{t+1}^o)\sim \frac{c_y}{c_z}(e_1G^t,e_2G^{t+1})=\frac{c_y}{c_z}(1,Gw), \label{eq:Rf}
\end{equation}
where $w\coloneqq e_2/e_1$ is the old to young income ratio and we have used the homogeneity of $c$ (Assumption \ref{asmp:U}\ref{item:U_c}). Obviously, for the fundamental value of housing to be finite, the interest rate cannot fall below the rent growth rate $G^\gamma$ in the long run. This heuristic argument motivates the following (non)existence result.

\begin{thm}[(Non)existence of fundamental equilibria]\label{thm:gamma<1f}
Suppose Assumptions \ref{asmp:G}--\ref{asmp:c} hold, $\gamma<1$, and let $w=e_2/e_1$. Then the following statements are true.
\begin{enumerate}
    \item\label{item:f_w*f} There exists a unique $w_f^*>0$ satisfying 
    \begin{equation}
        \frac{c_y}{c_z}(1,Gw_f^*)=G^\gamma. \label{eq:w*f}
    \end{equation}
    \item\label{item:f_exist} If $w>w_f^*$, there exists a fundamental long-run equilibrium. The equilibrium objects have the order of magnitude
    \begin{subequations}\label{eq:eqobj_gamma<1f}
        \begin{align}
            (c_t^y,c_t^o)&\sim (e_1G^t,e_2G^t), \label{eq:yz<1f}\\
            (P_t,r_t)&\sim \left(me_1^\gamma\frac{G^\gamma c_z}{c_y-G^\gamma c_z}\frac{c^\gamma}{c_y}G^{\gamma t},me_1^\gamma \frac{c^\gamma}{c_y}G^{\gamma t}\right), \label{eq:Pr<1f}\\
            R_t&\sim \frac{c_y}{c_z}>G^\gamma, \label{eq:R<1f}
        \end{align}
    \end{subequations}
    where $c,c_y,c_z$ are evaluated at $(y,z)=(1,Gw)$.
    \item\label{item:f_nonexist} If $w<w_f^*$, there exist no fundamental equilibria. All equilibria are bubbly with $\liminf_{t\to\infty} G^{-t}P_t>0$.
\end{enumerate}
\end{thm}

Although the conclusion that fundamental equilibria may fail to exist (unlike in pure bubble models, in which fundamental equilibria always exist) is surprising, its intuition is actually straightforward. As discussed above, in any fundamental equilibrium, the consumption allocation is asymptotically autarkic and the interest rate is pinned down as the marginal rate of intertemporal substitution evaluated at the autarkic allocation. Hence the order of magnitude \eqref{eq:eqobj_gamma<1f} immediately follows from the general analysis in Theorem \ref{thm:eq}. Because both the housing price and rent grow at rate $G^\gamma$, the interest rate (which equals the return on housing by no-arbitrage) must exceed $G^\gamma$ as in \eqref{eq:R<1f}. Hence, the no-bubble condition holds and the housing price just reflects the fundamentals. As the young to old income ratio $1/w=e_1/e_2$ rises, the autarkic interest rate falls. But it cannot fall below the rent growth rate $G^\gamma$, for otherwise the fundamental value would become infinite, which is impossible in equilibrium. Therefore, there cannot be any fundamental equilibria if the young are sufficiently rich. The threshold for the nonexistence of fundamental equilibria is determined by equating the marginal rate of intertemporal substitution to the rent growth rate $G^\gamma$, which is precisely the condition \eqref{eq:w*f}.

It is important to recognize the differences in statements \ref{item:f_exist} and \ref{item:f_nonexist}. All statement \ref{item:f_exist} claims is that there exists a fundamental long-run equilibrium satisfying the order of magnitude \eqref{eq:eqobj_gamma<1f}. It does not rule out the possibility that there are other equilibria that are potentially cyclic or chaotic. In contrast, statement \ref{item:f_nonexist} is much stronger. Under the condition $w<w_f^*$, it claims that no fundamental equilibria can exist at all, regardless of the asymptotic behavior such as convergent, cyclic, or chaotic. The proof of Theorem \ref{thm:gamma<1f}\ref{item:f_nonexist} is an application of the Bubble Necessity Theorem (Lemma \ref{lem:neccesity}), where the long-run rent growth rate established in Theorem \ref{thm:Gr} plays a crucial role.

\subsection{Existence of bubbly equilibria}

Theorem \ref{thm:gamma<1f} establishes a necessary and sufficient condition for the existence of a fundamental equilibrium. In particular, if the young are sufficiently rich and $w<w_f^*$, fundamental equilibria do not exist and hence bubbles are inevitable. The following theorem provides a necessary and sufficient condition for the existence of a bubbly long-run equilibrium.

\begin{thm}[Existence of bubbly long-run equilibrium]\label{thm:gamma<1b}
Suppose Assumptions \ref{asmp:G}--\ref{asmp:c} hold, $\gamma<1$, and let $w=e_2/e_1$. Then the following statements are true.
\begin{enumerate}
    \item\label{item:steady<1b} There exists a unique $w_b^*>w_f^*$ satisfying 
    \begin{equation}
        \frac{c_y}{c_z}(1,Gw_b^*)=G, \label{eq:w*b}
    \end{equation}
    which depends only on $G$ and $c$. A bubbly steady state of the system \eqref{eq:H} exists if and only if $w<w_b^*$, which is uniquely given by $\xi_b^*=(s^*,0)$ with $s^*=\frac{w_b^*-w}{w_b^*+1}$.
    \item\label{item:order<1b} For generic $G>1$ and $w<w_b^*$, there exists a bubbly long-run equilibrium. The equilibrium objects have the order of magnitude
    \begin{subequations}\label{eq:eqobj_gamma<1}
        \begin{align}
            (c_t^y,c_t^o)&\sim ((1-s^*)e_1G^t,(w+s^*)e_1G^t), \label{eq:yz<1b}\\
            (P_t,r_t)&\sim \left(s^*e_1G^t, me_1^\gamma \frac{c^\gamma}{c_y}G^{\gamma t}\right), \label{eq:Pr<1b}\\
            R_t&\sim G, \label{eq:R<1b}
        \end{align}
    \end{subequations}
    where $c,c_y$ are evaluated at $(y,z)=(1-s^*,G(w+s^*))$.
\end{enumerate}
\end{thm}

We explain the intuition for the following points:
\begin{enumerate*}
    \item\label{item:R=G} Why does the bubbly equilibrium interest rate $R$ equal the economic growth rate $G$?
    \item\label{item:w*b} Why do the young need to be sufficiently rich for the emergence of bubbles?
    \item\label{item:gamma<1} Why is the condition $\gamma<1$ important for the emergence of bubbles?
\end{enumerate*}
The intuition for \ref{item:R=G} is the following. In order for a housing bubble to exist in the long run, housing price must asymptotically grow at the same rate $G$ as the economy as in \eqref{eq:Pr<1b}: clearly housing price cannot grow faster than $G$ (otherwise the young cannot afford housing); if it grows at a lower rate than $G$, housing becomes asymptotically irrelevant. Because housing price grows at rate $G$ but the rent grows at rate $G^\gamma<G$, the interest rate \eqref{eq:R} must converge to $G$ as in \eqref{eq:R<1b}. The intuition for \ref{item:w*b} is the following. With bubbles, we know $R=G$. Because the young are saving through the purchase of housing, the lowest possible interest rate in the economy is the autarkic interest rate. Therefore, for the emergence of bubbles, the autarkic interest rate must be lower than the economic growth rate, or equivalently the young must be sufficiently rich. The condition \eqref{eq:w*b}, which equates the marginal rate of intertemporal substitution to the growth rate (long-run interest rate), determines the income ratio threshold for which such a situation is possible. The intuition for \ref{item:gamma<1} is the following. With bubbles, we know $R=G$ and the housing price grows at the same rate. Then the no-arbitrage condition \eqref{eq:R} forces the rents relative to the prices to be negligible (grow slower), for otherwise the interest rate will exceed the housing price growth rate and there will be no bubbles. Thus for the emergence of bubbles, we need $G>G^\gamma$ and hence $\gamma<1$.

In this bubbly equilibrium, the housing expenditure $S_t$ and rent $r_t$ asymptotically grow at rates $G$ and $G^\gamma<G$, respectively. On the other hand, since the gross risk-free rate \eqref{eq:R<1b} converges to $G$ and the rent grows at rate $G^\gamma<G$, the present value of rents---the fundamental value of housing $V_t$ in \eqref{eq:Vt}---is finite and grows at rate $G^\gamma$. Then the ratio $S_t/V_t$ grows at rate $G^{1-\gamma}>1$, so the housing price eventually exceeds the fundamental value and there is a bubble. Moreover, from a backward induction argument, we will have housing bubbles at all dates.

In the bubbly equilibrium, the housing price grows faster than the rent and is disconnected from fundamentals in the sense that the housing price is asymptotically independent of the preferences for housing. To see this, note that the threshold $w_b^*$ in \eqref{eq:w*b} depends only on the growth rate $G$ and the utility of consumption $c$. Then the steady state $s^*$ depends only on $G$, $c$, and incomes $(e_1,e_2)$, and so does the asymptotic housing price in \eqref{eq:Pr<1b}. In particular, the housing price is asymptotically independent of the marginal utility of housing $m$ as well as the elasticity of substitution $1/\gamma$ between consumption and housing. In contrast, the rent in \eqref{eq:Pr<1b} does depend on these parameters.

\subsection{Uniqueness of equilibria}

Although it is natural to focus on equilibria converging to steady states (\ie, long-run equilibria), there may be other equilibria. In general, an equilibrium is called \emph{locally determinate} if there are no other equilibria in a neighborhood of the given equilibrium. If a model does not make determinate predictions, its value as a tool for economic analysis is severely limited \citep{KehoeLevine1985}. Therefore, local determinacy of equilibrium is crucial for applications.

It is well known that equilibria in Arrow-Debreu economies are generically locally determinate \citep{Debreu1970} but not necessarily so in OLG models \citep{Gale1973,GeanakoplosPolemarchakis1991}. In our context, local determinacy means that there are no other equilibria converging to the same steady state. However, we already know the uniqueness of steady states, and we also know that Lemma \ref{lem:backward} allows us to establish global properties of equilibrium. Thus in our model, local determinacy implies equilibrium uniqueness, which justifies comparative statics and dynamics.

The local determinacy of a dynamic general equilibrium model often depends on the elasticity of intertemporal substitution (EIS) defined by
\begin{equation}
    \varepsilon(y,z)=-\left(\frac{\diff \log(c_y/c_z)}{\diff \log (y/z)}\right)^{-1}; \label{eq:EIS_def}
\end{equation}
see the discussion in \citet{FlynnSchmidtToda2023TE}. When $c$ is homogeneous of degree 1, we can show that $\varepsilon=\frac{c_yc_z}{cc_{yz}}$ (Lemma \ref{lem:c}). The following proposition provides a sufficient condition for the uniqueness of equilibria.

\begin{prop}[Uniqueness of equilibria]\label{prop:unique}
Suppose Assumptions \ref{asmp:G}--\ref{asmp:c} hold and $\gamma<1$. Let $w=e_2/e_1$ and $w_f^*,w_b^*$ be as in \eqref{eq:w*f} and \eqref{eq:w*b}. Then the following statements are true.
\begin{enumerate}
    \item\label{item:locdet<1f} If $w>w_f^*$, there exists a unique fundamental long-run equilibrium.
    \item\label{item:locdet<1b} If $w<w_b^*$ and the elasticity of intertemporal substitution \eqref{eq:EIS_def} satisfies
    \begin{equation}
        \frac{1-w_b^*}{2}\frac{1-w/w_b^*}{1+w}<\varepsilon(y,z)\neq \frac{1-w/w_b^*}{1+w} \label{eq:locdet_gamma<1}
    \end{equation}
    at $(y,z)=(1-s^*,G(w+s^*))$ with $s^*=\frac{w_b^*-w}{w_b^*+1}$, then there exists a unique bubbly long-run equilibrium.
\end{enumerate}
\end{prop}

Theorem \ref{thm:gamma<1f}\ref{item:f_exist} shows that all fundamental long-run equilibria are asymptotically equivalent. Proposition \ref{prop:unique}\ref{item:locdet<1f} shows that the fundamental equilibrium is actually unique. The right-hand side of \eqref{eq:locdet_gamma<1} is less than 1 because $0<w<w_b^*$. Therefore, the left-hand side of \eqref{eq:locdet_gamma<1} is less than $1/2$. Proposition \ref{prop:unique}\ref{item:locdet<1b} thus states that the bubbly equilibrium in Theorem \ref{thm:gamma<1b}\ref{item:order<1b} is locally determinate as long as the elasticity of intertemporal substitution (EIS) is not too much below $1/2$.\footnote{In general equilibrium theory, it is well known that multiple equilibria are possible if the elasticity is low; see \citet{TodaWalsh2017ETB} for concrete examples and \citet{TodaWalsh2024JME} for a recent review.}

The intuition for Proposition \ref{prop:unique} is as follows. Whether the bubbly equilibrium is locally determinate or not depends on the stability of linearized system around the steady state $\xi_b^*$. It turns out that one eigenvalue is $\lambda_2\coloneqq G^{\gamma-1}<1$, which is stable. The other eigenvalue $\lambda_1$ could be greater than 1 in modulus (unstable) or less (stable), depending on the model parameters. We find that as long as the EIS is not too much below $1/2$ (namely the left inequality of \eqref{eq:locdet_gamma<1} holds) and is distinct from the special value in the right-hand side of \eqref{eq:locdet_gamma<1} (in which case linearization is inapplicable due to a singularity), then $\abs{\lambda_1}>1$ (unstable). Since the dynamics has one endogenous initial condition (because $\xi_0=(s_0,e_1^{\gamma-1})$ and the initial young income $e_1$ is exogenous), the equilibrium is locally determinate: there exists a unique equilibrium path converging to the steady state if $e_1$ is large enough. Then the existence and uniqueness of equilibrium with arbitrary $e_1$ follows from the backward induction argument in Lemma \ref{lem:backward}. The same argument applies to the fundamental equilibrium, although in this case we have $\lambda_1>1$ regardless of the EIS.

\section{Possibility, necessity, and phase transition}\label{sec:phase}

Having established the existence and determinacy of equilibria, in this section we further develop the intuition, discuss expectation-driven housing bubbles, and present comparative dynamics exercises using a numerical example.

\subsection{Two-stage phase transition along economic development}\label{subsec:two-stage}

Theorems \ref{thm:gamma<1f} and \ref{thm:gamma<1b} imply that, as the young (more precisely, home buyers) become richer, the economy experiences \emph{two} phase transitions in the process of economic development, as illustrated in Figure \ref{fig:regime}, which shows how the elasticity of substitution between consumption and housing service $1/\gamma$ and young to old income ratio $1/w=e_1/e_2$ affect the equilibrium housing price regimes. (The case $1/\gamma\le 1$ is treated in Appendix \ref{sec:gamma>=1}.) We capture economic development with changes in the long-run income ratio of the young (home buyers) relative to the old (home sellers).

\begin{figure}[htb!]
\centering

\begin{tikzpicture}[scale = 2.8]

\draw[->] (-0.1,0) -- (4.03,0) node[right] {$1/\gamma$};
\draw (2,0) node[below] {Elasticity of substitution};
\draw[->] (0,-0.1) -- (0,2.5) node[above] {$1/w=e_1/e_2$};
\draw (0,0) node[below left] {$0$};

\fill[semitransparent,yellow] (0,0) rectangle (1,2.47);
\fill[semitransparent,red] (1,1) rectangle (4,2.47); 
\fill[nearly transparent,blue,domain = 1:4,variable = \x]
    (1,0)
    -- plot (\x,{1.5^(-1/\x + 1)})
    -- (4,0)
    -- cycle; 

\draw (1,0) -- (1,2.47);
\draw (1,0) node[below] {$1$};

\draw[domain = 1:4,dashed] plot (\x,{1.5^(-1/\x + 1)});
\draw (4,1.355) node[right] {$1/w_f^*$};

\draw[dashed] (1,1) -- (4,1);
\draw (4,1) node[right] {$1/w_b^*$};

\draw (0,2.47) node[below right] {(Pathological)};
\draw (0,5/4) node[right] {Fundamental};
\draw (0,0) node[above right] {(Prop.~\ref{prop:gamma>1})};
\draw[->] (1.1,2.6) -- (1.01,2.51);
\draw (1.1,2.47) node[above right] {Fundamental (Prop.~\ref{prop:gamma=1})};
\draw (4,0) node[above left] {Fundamental regime};
\draw (1,0) node[above right] {(Thm.~\ref{thm:gamma<1f})};
\draw (4,1) node[above left] {Bubble possibility regime};
\draw (4,2.47) node[below left] {Bubble necessity regime};
\draw (1,2.47) node[below right] {(Thm.~\ref{thm:gamma<1b})};

\draw (0,5/4) node[left] {Young};
\draw[->] (-0.2,1.35) -- (-0.2,1.5) node[above] {rich};
\draw[->] (-0.2,1.15) -- (-0.2,1) node[below] {poor};
    
\end{tikzpicture}

\caption{Phase transition of equilibrium housing price regimes.}\label{fig:regime}

\caption*{\footnotesize Note: $1/w=e_1/e_2$ is the young to old income ratio and $w_f^*,w_b^*$ are the thresholds for the bubble necessity and possibility regimes defined by \eqref{eq:w*f} and \eqref{eq:w*b}, respectively. The figure corresponds to the CES utility \eqref{eq:CES} with $\beta=1/2$, $\sigma=1$, and $G=1.5$.}

\end{figure}

When the young to old income ratio $1/w=e_1/e_2$ is below the bubbly equilibrium threshold $1/w_b^*$, the young do not have sufficient purchasing power to drive up the housing price and only fundamental equilibria exist (Theorem \ref{thm:gamma<1f}\ref{item:f_exist}). In this fundamental regime, the housing price grows at rate $G^\gamma$, which is lower than both the interest rate $R$ and the economic growth rate $G$. In the long run, the expenditure share of housing converges to zero and the consumption allocation becomes autarkic (see \eqref{eq:yz<1f}).

When the income ratio of the young exceeds the first critical value $1/w_b^*$, the economy transitions to the bubble possibility regime in which fundamental and bubbly equilibria coexist (Theorem \ref{thm:gamma<1b}). In this regime, although each equilibria are determinate, which equilibrium will be selected depends on agents' expectations.

When the income ratio of the young exceeds the second and still higher critical value $1/w_f^*$, fundamental equilibria cease to exist and all equilibria become bubbly (Theorem \ref{thm:gamma<1f}\ref{item:f_nonexist}). Bubbles are necessary for the existence of equilibrium and the bubble necessity regime emerges. In this regime, the housing price is asymptotically determined only by the economic growth rate $G$ and the preference for consumption goods $c$, and thus the housing price inevitably becomes disconnected from fundamentals.

The intuition for the necessity of housing bubbles when the young are sufficiently rich is the following. As discussed above, in any fundamental equilibrium, the expenditure share of housing converges to zero and the consumption allocation becomes autarkic. However, as the young get richer (the young to old income ratio $1/w$ increases), the interest rate $R=(c_y/c_z)(1,Gw)$ falls (Figure \ref{fig:interest}). If $R$ gets lower than a critical value, the economy enters the bubble possibility regime. Hence, housing bubbles driven by optimistic expectations may be possible. As the income ratio increases further, the fundamental equilibrium interest rate becomes lower than the rent growth rate $G^\gamma$. If the economy enters that situation, the only possible equilibrium is one that features a housing bubble.

\begin{figure}[htb!]
\centering

\begin{tikzpicture}[scale = 2.8]

\draw[->] (-0.1,0) -- (4.03,0) node[below] {$1/w=e_1/e_2$};
\draw[->] (0,-0.1) -- (0,2.5) node[left] {$R$};
\draw (0,0) node[below left] {$0$};

\fill[semitransparent,green] (0,1.5) rectangle (4,2.47);
\fill[semitransparent,lightgray] (0,0) rectangle (4,1.5);

\draw[dashed] (1,0) node[below] {$1/w_b^*$} -- (1,2.47);
\draw[dashed] (2,0) node[below] {$1/w_f^*$} -- (2,2.47);


\draw (4,0) node[above left] {poor $\leftarrow$ Young $\rightarrow$ rich};

\draw[very thick,red] (1,1.5) -- (4,1.5);
\draw (2,1.5) node[above right,text=red] {Bubbly equilibrium};
\draw[dashed] (0,1.5) -- (1,1.5);
\draw (0,1.5) node[left] {$G$};

\draw[domain = 0.607:2,very thick,blue] plot (\x,{1.5/\x});
\draw (2,0.75) node[right,text=blue] {Fundamental equilibrium};
\draw[dashed] (0,0.75) -- (2,0.75);
\draw (0,0.75) node[left] {$G^\gamma$};

\draw [decorate,
    decoration = {brace}, thick] (0.05,2.5) --  (0.95,2.5)
    node[midway,above]{Fundamental};
\draw [decorate,
    decoration = {brace}, thick] (1.05,2.5) --  (1.95,2.5)
    node[midway,above]{Possibility};
\draw [decorate,
    decoration = {brace}, thick] (2.05,2.5) --  (3.95,2.5)
    node[midway,above]{Bubble necessity};

\end{tikzpicture}
\caption{Housing price regimes and equilibrium interest rate.}\label{fig:interest}

\caption*{\footnotesize Note: see Figure \ref{fig:regime} for explanation of parameters.}
\end{figure}

Furthermore, we emphasize that once the state of the economy changes to the housing bubble economy, whether by expectations or by necessity, the determination of housing prices becomes purely demand-driven: the housing price continues to rise due to sustained demand growth arising from income growth of the young (home buyers). In contrast, when housing prices reflect fundamentals, it equals the present value of housing rents and hence its determination is supply-driven. The demand-driven housing price dynamics is a distinctive feature of the housing bubble economy.

We would like to add an important remark concerning the knife-edge case with $\gamma=1$, \ie, the Cobb-Douglas case, which is often employed in housing models or macroeconomic analyses \citep{Kocherlakota2009,ArceLopez-Salido2011}. When $\gamma=1$, steady-state (balanced) growth emerges, in which case housing rents and prices grow at the same rate and therefore housing bubbles are impossible. This result has critically important implications for the method of macroeconomic modeling. As long as we construct a model so that only steady-state growth with stationarity emerges, by model construction, housing bubbles can never occur. What our analyses show is that once we deviate from the knife-edge restriction,\footnote{As is well known as the ``Uzawa steady-state (balanced) growth theorem'' \citep{Uzawa1961}, any growth model that produces a balanced growth path is a knife-edge theory. Indeed, \citet*[p.~1306]{GrossmanHelpmanOberfieldSampson2017} clearly note ``As with any model that generates balanced growth, knife-edge restrictions are required to maintain the balance''.} asset pricing implications become markedly different. This implies that the essence of housing bubbles is nonstationarity. (See also the introduction and concluding remarks in \citet{HiranoToda2024JME}.)

\subsection{Expectation-driven housing bubbles along economic development} \label{subsec:expectation}

We illustrate the preceding analysis and the role of expectations with a numerical example. Suppose the composite consumption takes the CES form \eqref{eq:CES}. A straightforward calculation yields
\begin{equation}
    c_y=(1-\beta)(y/c)^{-\sigma} \quad \text{and} \quad c_z=\beta (z/c)^{-\sigma}. \label{eq:CES_grad}
\end{equation}
Using \eqref{eq:w*f}, \eqref{eq:w*b}, and \eqref{eq:CES_grad}, we can solve for the critical values for the existence of fundamental and bubbly equilibria as
\begin{subequations}
\begin{align}
    \frac{1-\beta}{\beta}(Gw_f^*)^\sigma=G^\gamma &\iff w_f^*=\left(\frac{\beta}{1-\beta}G^{\gamma-\sigma}\right)^{1/\sigma}, \label{eq:w*f_CES} \\
    \frac{1-\beta}{\beta}(Gw_b^*)^\sigma=G &\iff w_b^*=\left(\frac{\beta}{1-\beta}G^{1-\sigma}\right)^{1/\sigma}. \label{eq:w*b_CES}
\end{align}
\end{subequations}
Substituting \eqref{eq:CES_grad} into \eqref{eq:s_dynamics2}, we obtain
\begin{equation}
    \beta S_{t+1}z^{-\sigma}=(1-\beta)S_ty^{-\sigma}-m c^{\gamma-\sigma}, \label{eq:s_dynamics_CES}
\end{equation}
where $(y,z)=(e_t^y-S_t,e_{t+1}^o+S_{t+1})$. To solve for the equilibrium numerically, we can take a large enough $T$, set $S_T=s^*e_T^y$ with steady state value $s^*$ defined by
\begin{equation*}
    s^*=\begin{cases*}
        0 & if fundamental equilibrium,\\
        \frac{w_b^*-w}{w_b^*+1} & if bubbly equilibrium,
    \end{cases*}
\end{equation*}
and solve the nonlinear equation \eqref{eq:s_dynamics_CES} backwards for $S_{T-1},\dots,S_0$. Note that the backward calculations of $\set{S_t}_{t=0}^T$ are always possible by Lemma \ref{lem:backward}.

As a numerical example, we set $\beta=1/2$, $\sigma=1$, $\gamma=1/2$, $m=0.1$, and $G=1.1$. The income ratio threshold for the bubble possibility regime \eqref{eq:w*b_CES} is then $w_b^*=1$. Figure \ref{fig:dynamics_f} shows the equilibrium housing price dynamics when $(e_1,e_2)=(95,105)$ so that $e_2/e_1>w_b^*$ and hence only a fundamental equilibrium exists. The housing price and rent asymptotically grow at the same rate $G^\gamma$, which is lower than the endowment growth rate $G$. Furthermore, the distance in semilog scale between the housing price and rent converges, suggesting that the price-rent ratio converges. These observations are consistent with Theorem \ref{thm:gamma<1f}.

\begin{figure}[htb!]
    \centering
    \begin{subfigure}{0.48\linewidth}
    \includegraphics[width=\linewidth]{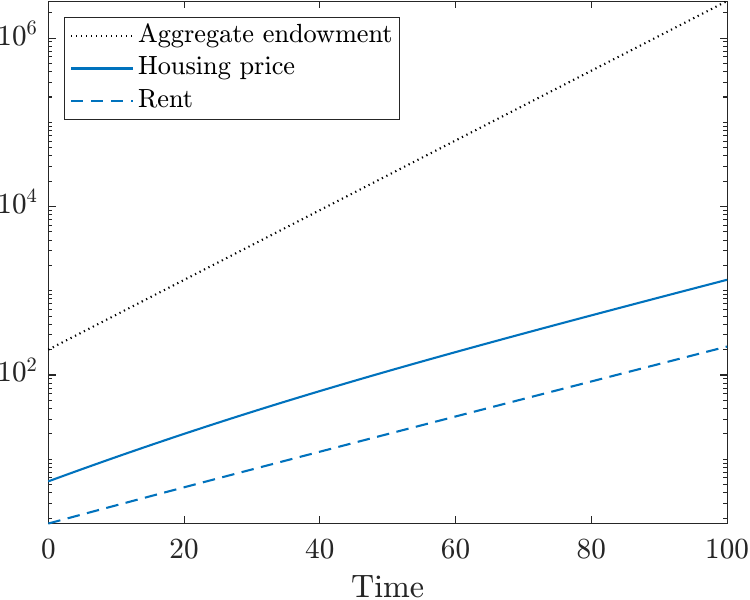}
    \caption{Fundamental equilibrium.}\label{fig:dynamics_f}
    \end{subfigure}
    \begin{subfigure}{0.48\linewidth}
    \includegraphics[width=\linewidth]{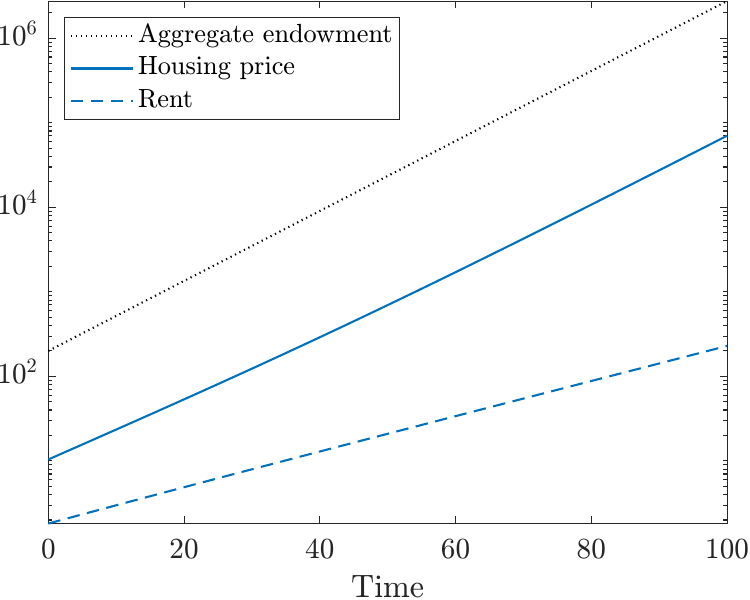}
    \caption{Bubbly equilibrium.}\label{fig:dynamics_b}
    \end{subfigure}
    \caption{Equilibrium housing price dynamics.}\label{fig:dynamics}
\end{figure}

Figure \ref{fig:dynamics_b} repeats the same exercise for $(e_1,e_2)=(105,95)$ so that $e_2/e_1<w_b^*$ and a bubbly equilibrium exists. The housing price asymptotically grows at the same rate as endowments, while the rent grows at a slower rate. Consequently, the price-rent ratio increases. These observations are consistent with Theorem \ref{thm:gamma<1b}.

We next study how expectations about economic development and incomes in the future affect the current housing price. In Figure \ref{fig:dynamics_fbf_I0}, we consider phase transitions between the fundamental and bubbly regimes. The economy starts with $(e_0^y,e_0^o)=(95,105)$ and agents believe that the endowments grow at rate $G$ and the income ratio $e_t^o/e_t^y$ is constant at $105/95$. At $t=40$, the income ratio $e_t^o/e_t^y$ unexpectedly changes to $95/105$ and agents believe that this new ratio will persist. Thus the economy takes off to the bubbly regime. Finally, at $t=80$ the income ratio $e_t^o/e_t^y$ unexpectedly reverts to the original value $105/95$. Note that as the economy enters the bubbly regime, rents are hardly affected but the housing price increases and grows at a faster rate, generating a housing bubble.

\begin{figure}[htb!]
    \centering
    \begin{subfigure}{0.48\linewidth}
    \includegraphics[width=\linewidth]{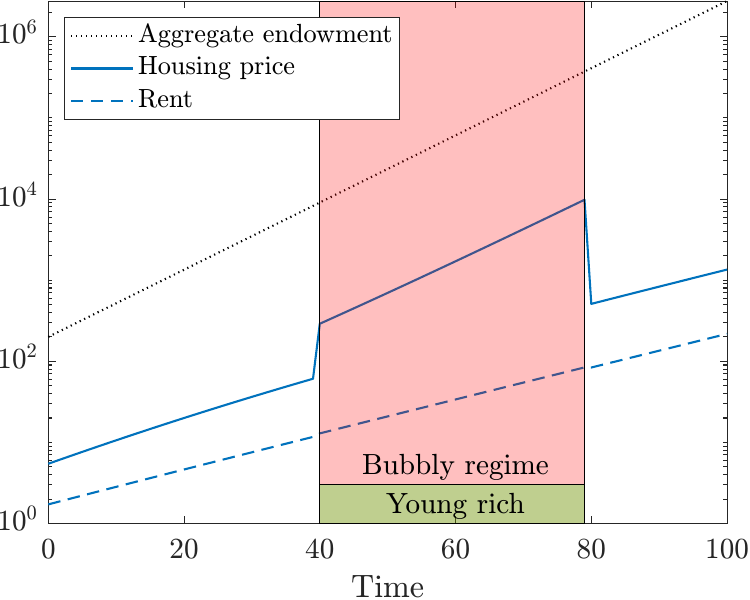}
    \caption{Unexpected income change.}\label{fig:dynamics_fbf_I0}
    \end{subfigure}
    \begin{subfigure}{0.48\linewidth}
    \includegraphics[width=\linewidth]{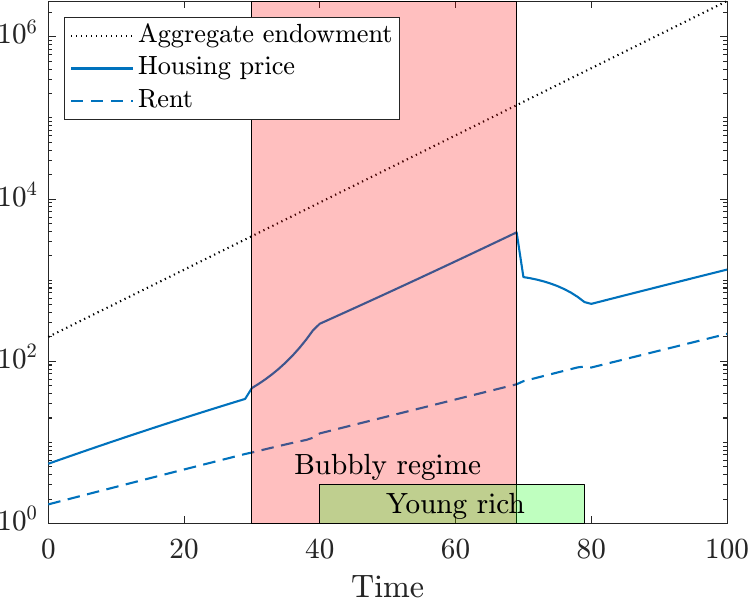}
    \caption{Expected income change.}\label{fig:dynamics_fbf_I10}
    \end{subfigure}
    \caption{Phase transition between fundamental and bubbly regimes.}\label{fig:dynamics_fbf}
\end{figure}

Figure \ref{fig:dynamics_fbf_I10} repeats the same exercise except that the income changes are anticipated. Specifically, agents learn at $t=30$ that the income ratio will change to $95/105$ (so the young will be relatively rich) starting at $t=40$ and will remain so forever. Similarly, agents learn at $t=70$ that the income ratio will revert to $105/95$ (so the young will be relatively poor) starting at $t=80$ and will remain so forever. In this case, the economy takes off to the bubbly regime at $t=30$ and reenters the fundamental regime at $t=70$ due to rational expectations. We can see that the housing price jumps up at $t=30$ and grows fast even before the fundamentals change. The housing price already contains a bubble, even if the current income of the young is relatively low and appears to be incapable of generating bubbles. This is due to a backward induction argument: if there is a bubble in the future (so \eqref{eq:TVC} holds with strict inequality and the no-bubble condition fails), there is a bubble in every period. Once the young become relatively rich at $t=40$, the housing price increases at the same rate as endowments, consistent with Theorem \ref{thm:gamma<1b}. The housing bubble collapses at $t=70$ when agents learn that the young will be relatively poor in the future, even though the young remain relatively rich until $t=80$. 

From this analysis, we can draw an interesting implication. During expectation-driven housing bubbles, housing prices grow faster than rents. The price-income ratio continues to rise and hence the dynamics may appear unsustainable. Moreover, the greater the time gap between when news of rising incomes arrives ($t=30$) and when incomes actually start to rise ($t=40$), the longer the duration of the seemingly unsustainable dynamics. This expectation-driven housing bubbles and their collapse may capture realistic transitional dynamics. For instance, \citet{MilesMonro2021} emphasize that the decline in the real interest rate has produced large effects on the evolution of housing prices in the U.K. In our model, the (real) interest rate is endogenously determined and is closely related to the income of home buyers. As their income rises and the interest rate falls below the rent growth rate, a housing bubble necessarily emerges. \citet{MankiwWeil1989} and \citet*{KiyotakiMichaelidesNikolov2011,KiyotakiMichaelidesNikolov2024} stress the importance of expectation formation of long-run aggregate income growth and the interest rate to account for the fluctuations in housing prices. Our expectation-driven housing bubbles and their collapse show that even small changes in incomes of home buyers or the expectation thereof could produce large swings in housing prices. A critical difference is that housing prices in their papers reflect fundamentals, while our main focus is to identify the economic conditions under which housing prices reflect fundamentals or contain bubbles and to study expectation-driven housing price bubbles.

\section{Discussion and extensions}\label{sec:extension}

\subsection{Multiple savings vehicles}\label{subsec:extension_multiple}

By Theorem \ref{thm:eq}, an equilibrium always exists. By Theorem \ref{thm:gamma<1f}\ref{item:f_nonexist}, fundamental equilibria do not exist when $w<w_f^*$. Therefore, under this condition all equilibria are bubbly. In our model, housing is the only financial asset. A natural question would be what happens with multiple savings vehicles. To address this issue, take any (bubbly) equilibrium with housing price $P_t=V_t+B_t$, where $V_t$ is the fundamental value \eqref{eq:Vt} and $B_t\coloneqq P_t-V_t\ge 0$ is the bubble component. By the definitions of the interest rate \eqref{eq:R} and the fundamental value \eqref{eq:Vt}, we obtain
\begin{align*}
    P_t&=\frac{1}{R_t}(P_{t+1}+r_{t+1}), & V_t&=\frac{1}{R_t}(V_{t+1}+r_{t+1}).
\end{align*}
Taking the difference, we obtain $B_{t+1}=R_tB_t$, so the bubble component grows at the rate of interest, as is well known. Now, take any $\theta\in [0,1]$ and define an alternative housing price by $P_t'=V_t+(1-\theta) B_t$. Then we can easily construct an equilibrium in which the housing price is $P_t'$ and there is an additional pure bubble asset (intrinsically worthless asset that pays no dividends) with market capitalization $\theta B_t$. This argument (which is the same as the ``bubble substitution'' argument in \citet[\S5]{Tirole1985}) shows that once there are multiple assets, the size of the bubble attached to each individual asset may become indeterminate (here parametrized by $\theta\in [0,1]$) because all assets are perfect substitute. However, the total size of the bubble is determinate (equal to $B_t$) and hence the consumption allocation as well as all macroeconomic implications are identical regardless of $\theta$.\footnote{Although the model is rather different, \citet*{HiranoJinnaiTodaLeverage} develop a macro-finance model in which there are multiple savings vehicles including capital, land, and bonds, and show that land price bubbles necessarily emerge when the financial leverage gets sufficiently high.} Hence, from a macroeconomic perspective, this indeterminacy is unimportant. This result is different from standard pure bubble models, where equilibria exhibit real indeterminacy \citep{Gale1973,HiranoToda2024EL}.

\subsection{Welfare implications}\label{subsec:extension_welfare}

In \S\ref{sec:longrun}, we saw that housing bubbles can or must emerge as the young get richer. A natural question is whether housing bubbles are socially desirable or not. In this section we discuss the welfare implications of housing bubbles.

Let $\set{(c_t^y,c_t^o,h_t)}_{t=0}^\infty$ be an arbitrary allocation with $c_t^y,c_t^o>0$ and $c_t^y+c_t^o=e_t^y+e_t^o$. Since only the young have preference for housing service, which is perishable, it is obviously efficient to assign all housing service to the young. Using Assumption \ref{asmp:U}, the utility of generation $t$ becomes $U(c_t^y,c_{t+1}^o,1)=u(c(c_t^y,c_{t+1}^o))+mu(1)$, which is a monotonic transformation of $c(c_t^y,c_{t+1}^o)$. Therefore, the welfare analysis (in terms of Pareto efficiency) reduces to that of an endowment economy without housing and with utility function $c(y,z)$ for goods.

Let $G_t=e_{t+1}^y/e_t^y$ be the growth rate of young income and $w_t=e_t^o/e_t^y$ be the old to young income ratio at time $t$. Let $s_t=1-c_t^y/e_t^y$ be the saving rate. Then the utility of generation $t$ becomes
\begin{equation*}
    c(c_t^y,c_{t+1}^o)=c(e_t^y(1-s_t),e_{t+1}^y(w_{t+1}+s_{t+1}))=e_t^yc(1-s_t,G_t(w_{t+1}+s_{t+1})),
\end{equation*}
which is a monotonic transformation of $c(1-s_t,G_t(w_{t+1}+s_{t+1}))$. This argument shows that the welfare analysis reduces to the case in which the time $t$ aggregate endowment is $1+w_t$, the utility function of generation $t$ is $u_t(y,z)\coloneqq c(y,G_tz)$, and the proposed allocation is $(y_t,z_t)=(1-s_t,w_t+s_t)$. Since Assumption \ref{asmp:G} implies that $w_t=e_t^o/e_t^y$ is constant for $t\ge T$, we can apply the characterization of Pareto efficiency in OLG models with bounded endowments provided by \citet{BalaskoShell1980}. We thus obtain the following proposition.

\begin{prop}[Characterization of equilibrium efficiency]\label{prop:efficient}
Suppose Assumptions \ref{asmp:G}--\ref{asmp:c} hold, $\gamma<1$, and let $w=e_2/e_1$. Then the following statements are true.
\begin{enumerate}
    \item\label{item:efficient_w>wb} If $w\ge w_b^*$, any equilibrium is efficient.
    \item\label{item:efficient_b} If $w<w_b^*$, any bubbly long-run equilibrium is efficient.
    \item\label{item:efficient_f} If $w<w_b^*$, any fundamental long-run equilibrium is inefficient.
\end{enumerate}
\end{prop}

Recalling that $w<w_b^*$ implies $R<G$ in the fundamental equilibrium (Figure \ref{fig:interest}), fundamental equilibria are inefficient whenever $R<G$. Therefore, in Figure \ref{fig:interest}, all equilibria in the green region (including the boundary) are efficient, whereas all equilibria in the gray region (excluding the boundary) are inefficient.

The intuition for the Pareto inefficiency of fundamental equilibria when $w<w_b^*$ is the following. In equilibrium, since endowments grow at rate $G$ and rents grow at rate $G^\gamma$ (Theorem \ref{thm:Gr}), if the housing price equals its fundamental value, it must also grow at rate $G^\gamma$. Since $G^\gamma<G$, the housing price is asymptotically negligible relative to endowments, so the equilibrium consumption becomes autarkic. Now when $w<w_b^*$, the young are richer, so the interest rate becomes so low that it is below the economic growth rate (see \eqref{eq:w*b}). Housing prices are too low to absorb savings desired by the young. In other words, housing is not serving as a means of savings with enough returns. In this situation if we consider a social contrivance such that for each large enough $t$ the young at time $t$ gives the old $\epsilon G^t$ of the good (hence the old at time $t+1$ receives $\epsilon G^{t+1}$ of the good), it is as if agents are able to save at rate $G$ higher than the interest rate, which improves welfare. Since this argument holds for all large enough $t$, we have a Pareto improvement, which implies the inefficiency of the fundamental equilibrium.

While the statements in Proposition \ref{prop:efficient}\ref{item:efficient_w>wb}\ref{item:efficient_b} are hardly surprising given the results of \citet{Diamond1965} and \citet{Tirole1985}, we note that statement \ref{item:efficient_f} that fundamental equilibria are inefficient in the bubble possibility regime may not be entirely obvious. In fact, as noted in the introduction, the well-known result of \citet{McCallum1987} shows that the introduction of a productive non-reproducible asset eliminates dynamic inefficiency in OLG models. Contrary to common understandings, this result is not necessarily true. This apparent conflict is due to the fact that \citet{McCallum1987} implicitly assumes steady state growth (see his discussion around Endnotes 20 and 21), which holds only for the knife-edge Cobb-Douglas case (which corresponds to $\gamma=1$ in our model, treated in Appendix \ref{subsec:gamma=1}). Once we consider the global parameter space with respect to $\gamma$, the region with dynamically inefficient equilibria always exists when $\gamma<1$ (the gray region in Figure \ref{fig:interest}). Furthermore, inefficient equilibria arise only in the intermediate region of the income ratio $e_2/e_1$, so there is a non-monotonic relationship between the income ratio and the existence of dynamically inefficient equilibria.

It is fair to say that there are diverse views on the welfare implications of asset price bubbles.\footnote{In our model, both the housing and rental markets are frictionless. Even if the rental market was missing, because agents are homogeneous, they end up being owner-occupants. In reality, rental markets could have frictions perhaps due to moral hazard issues. Then the welfare implications could change when agents are heterogeneous. For instance, \citet{GraczykPhan2021} consider an OLG model with housing, missing rental markets, and agents with heterogeneous incomes and find that housing bubbles hurt poor agents because they are priced out of the housing market. As such, welfare implications of housing bubbles may depend on the model details but the existence of the positive effect would survive even if we consider various extensions.} Although these are anecdotal, it is often noted that in Russia, people do not trust banks and government bonds because of the experience of the collapse of the Soviet Union and the default of Russian government bonds in 1998. Similarly, in the United Kingdom, there are concerns about the sustainability of pensions. These circumstances imply that there are not enough savings vehicles with high returns and instead, housing is an effective means of saving. In these situations, welfare would improve if the housing bubble raises housing yields. Proposition \ref{prop:efficient} captures the positive aspects of these housing bubbles.

\subsection{Credit-driven housing bubbles}\label{subsec:extension_credit}

In our model, because the young are homogeneous and the old exit the economy, there cannot be any borrowing or lending in equilibrium. In reality, housing is usually purchased using credit. To study the role of credit in generating housing bubbles in the simplest setting, we consider an open economy in which an external banking sector (\eg, foreign investors in mortgage-backed securities) provides exogenous credit.

To construct such a model, let $\set{(\tilde{e}_t^y,\tilde{e}_t^o)}_{t=0}^\infty$ be the endowment of some (closed) economy with corresponding equilibrium risk-free rate and housing expenditure $\set{(R_t,S_t)}_{t=0}^\infty$. 
Take any sequence $\set{\ell_t}_{t=0}^\infty$ such that $\ell_t\in [0,\tilde{e}_t^y)$ and define $e_0^o=\tilde{e}_0^o$ and $(e_t^y,e_{t+1}^o)=(\tilde{e}_t^y-\ell_t,\tilde{e}_{t+1}^o+R_t\ell_t)$ for $t\ge 0$. Then we can construct an equilibrium in which the endowment is $(e_t^y,e_t^o)$, the interest rate is $R_t$, the housing expenditure is $S_t$, and the external banking sector provides loan $\ell_t$ to the young at time $t$. We can see this as follows. At time $t$, the available funds of the young is $e_t^y+\ell_t=\tilde{e}_t^y$. At time $t+1$, because the old repay $R_t\ell_t$, the available funds is $e_{t+1}^o-R_t\ell_t=\tilde{e}_{t+1}^o$. Therefore, given the available funds and the interest rate $R_t$, it is optimal for the young to spend $S_t$ on housing, so we have an equilibrium.

Combining this argument with the analysis in \S\ref{sec:longrun}, even if the income share of the young $e_t^y/e_t^o$ is low and a bubbly equilibrium may not exist, if the young have access to sufficient credit, a housing bubble may emerge.

\begin{prop}\label{prop:credit}
Let everything be as in Theorem \ref{thm:gamma<1b} and suppose the banking sector is willing to lend $\ell_t=\ell G^t$ to the young. If the loan to income ratio satisfies
\begin{equation}
    w>\lambda\coloneqq \frac{\ell}{e_1}>\frac{w-w_b^*}{w_b^*+1}, \label{eq:lambda_cond}
\end{equation}
then there exists a bubbly long-run equilibrium. Under this condition, the housing price has order of magnitude
\begin{equation}
    P_t\sim e_1\left(\frac{w_b^*-w}{w_b^*+1}+\lambda\right)G^t=s^*e_1G^t+\ell_t, \label{eq:p_lambda}
\end{equation}
so credit increases the housing price one-for-one.
\end{prop}

Proposition \ref{prop:credit} has two implications. First, the fact that external credit may drive housing bubbles is at least consistent with some narratives during the U.S.\ housing boom in the early 2000s, including the famous remarks by \citet{Bernanke2005} on the ``global saving glut''. \citet{BertautDeMarcoKaminTryon2012} document that a substantial fraction of mortgages were financed through mortgage-backed securities purchased by European investors (``external banking sector'').  \citet{BarlevyFisher2021} document that the share of interest-only mortgages is correlated with the housing price growth rates across regions. Second, using \eqref{eq:p_lambda} and $G^\gamma<G$, by a similar calculation as in \eqref{eq:yz<1b}, the consumption of the young has the order of magnitude
\begin{equation*}
    c_t^y=e_t^y+\ell_t-P_t-r_t\sim e_1G^t+\ell_t-(s^*e_1G^t+\ell_t)=(1-s^*)e_1G^t,
\end{equation*}
which is independent of credit $\ell_t$. Therefore, once home buyers have access to sufficient credit such that a housing bubble emerges, increasing credit further ends up raising the housing price one-for-one with no real effect on the long-run consumption allocation and hence welfare. Note that in reality there are financing costs, so a housing bubble driven by excessive credit could hurt welfare. (See \citet{Barlevy2018} for a discussion of policy issues regarding bubbles.)

\section{Concluding remarks}\label{sec:conclude}

The theory of housing bubbles remains largely underdeveloped due to the fundamental difficulty of attaching bubbles to dividend-paying assets \citep{SantosWoodford1997}. In this paper, we have taken the first step towards building a theory of rational housing bubbles. We have presented a bare-bones model of housing bubbles with phase transitions that can be used as a stepping stone for a variety of applications. In concluding our article, we discuss directions for future research.

To analyze how equilibrium housing prices are determined in the process of economic development in a tractable way, we based our analysis on the classical overlapping generations model. However, a variety of generalizations are possible, including Bewley-type models with infinitely-lived agents as in \citet[\S5]{HiranoToda2025JPE}. We hope that our bare-bones model of housing bubbles will lead to a variety of extensions both in theoretical and quantitative analyses.

Our theoretical analysis also provides testable implications. First, from the analysis on the long-run behavior, housing bubbles are more likely to emerge if the incomes (or available funds through credit) of home buyers are higher or expected to be higher in the process of economic development. If the incomes of home buyers rise as economic development progresses, housing bubbles may naturally arise first by optimistic expectations, and then inevitably emerge as the optimistic fundamentals materialize. There is some empirical evidence consistent with this narrative. \citet{GyourkoMayerSinai2013} document that an increase in the high-income population in a metropolitan area is associated with high housing appreciation. The demographic structure could also be exploited to test our theory (\eg, improved longevity or early retirement make the old ``poorer''). Second, if there is a housing bubble on the long-run trend, rents grow at rate $G^\gamma$, whereas housing prices grow at rate $G$, implying that the price-rent ratio will rise. Hence, an upward trend in the price-rent ratio could be an indicator for housing bubbles. The findings of \citet[Fig.~1]{AmaralDohmenKohlSchularick2024} and \citet[Fig.~1]{Backer-PeralHazellMianYield} are consistent with this narrative, and the bubble detection literature \citep{PhillipsShi2020} could be applied. We hope that our theoretical framework may be useful for empirical researchers to investigate these issues further.

\appendix

\section{Proofs}\label{sec:proof}

\subsection{Proof of lemmas}

The following lemma lists a few implications of Assumption \ref{asmp:c} that will be repeatedly used.

\begin{lem}\label{lem:c}
Suppose Assumption \ref{asmp:c} holds and let $g(x)\coloneqq c(x,1)$. Then the following statements are true.
\begin{enumerate}
    \item The first partial derivatives of $c$ are given by
    \begin{subequations}\label{eq:c_partial1}
    \begin{align}
        c_y(y,z)&=g'(y/z)>0,\\
        c_z(y,z)&=g(y/z)-(y/z)g'(y/z)>0
    \end{align}
    \end{subequations}
    and are homogeneous of degree 0.
    \item The second partial derivatives are given by
    \begin{subequations}\label{eq:c_partial2}
    \begin{align}
        c_{yy}(y,z)&=\frac{1}{z}g''(y/z)<0, \\
        c_{yz}(y,z)&=-\frac{y}{z^2}g''(y/z)>0, \\
        c_{zz}(y,z)&=\frac{y^2}{z^3}g''(y/z)<0.
    \end{align}
    \end{subequations}
    \item Fixing $z>0$, the marginal rate of substitution $c_y/c_z$ is continuously differentiable and strictly decreasing in $y$ and has range $(0,\infty)$.
    \item The elasticity of intertemporal substitution is $\varepsilon(y,z)=\frac{c_yc_z}{cc_{yz}}>0$.
\end{enumerate}
\end{lem}

\begin{proof}
By definition, $g(x)=c(x,1)$. Therefore, $g'(x)=c_y(x,1)>0$ and $g''(x)=c_{yy}(x,1)<0$ by Assumption \ref{asmp:c}. Since $c$ is homogeneous of degree 1, we have $c(y,z)=zc(y/z,1)=zg(y/z)$. Then \eqref{eq:c_partial1} and \eqref{eq:c_partial2} are immediate by direct calculation.

Fixing $z>0$, define the marginal rate of substitution $M(y)=(c_y/c_z)(y,z)$. Then $M$ is continuously differentiable because $c$ is twice continuously differentiable and $c_y,c_z>0$. Since $c_y,c_z$ are homogeneous of degree 0, we have
\begin{equation}
    M(y)=\frac{c_y(y,z)}{c_z(y,z)}=\frac{c_y(y/z,1)}{c_z(1,z/y)}. \label{eq:My}
\end{equation}
Since $c_y,c_z>0$ and $c_{yy},c_{zz}<0$, the numerator (denominator) is positive and strictly decreasing (increasing) in $y$. Therefore, $M$ is strictly decreasing. Furthermore, since $c_y(0,z)=c_z(y,0)=\infty$, letting $y\downarrow 0$ and $y\uparrow \infty$ in \eqref{eq:My}, we obtain $M(0)=\infty$ and $M(\infty)=0$, so $M$ has range $(0,\infty)$.

Finally, we derive the elasticity of intertemporal substitution (EIS) $\varepsilon$. Since $c$ is homogeneous of degree 1, we have $c(\lambda y,\lambda z)=\lambda c(y,z)$. Differentiating both sides with respect to $\lambda$ and setting $\lambda=1$, we obtain
\begin{equation}
    yc_y+zc_z=c.\label{eq:c_Euler}
\end{equation}
Letting $\sigma=1/\varepsilon$ and $x=y/z$, by the chain rule we obtain
\begin{align*}
    \sigma&=-\frac{\partial \log (c_y/c_z)(xz,z)}{\partial \log x}=-x\frac{c_z}{c_y}\frac{zc_{yy}c_z-c_yzc_{yz}}{c_z^2}\\
    &=y\frac{c_yc_{yz}-c_zc_{yy}}{c_yc_z}=\frac{(yc_y+zc_z)c_{yz}}{c_yc_z}=\frac{cc_{yz}}{c_yc_z},
\end{align*}
where the last line uses \eqref{eq:c_partial2} and \eqref{eq:c_Euler}.
\end{proof}

\begin{proof}[Proof of Lemma \ref{lem:backward}]
Let $\cS_T=\set{S_t}_{t=T}^\infty$ be an equilibrium starting at $t=T$. Set $t=T-1$ and define the function $f:[0,e_{T-1}^y)\to \R$ by $f(S)=S_Tc_z-Sc_y+mc^\gamma$, where $c,c_y,c_z$ are evaluated at $(y,z)=(e_{T-1}^y-S,b_T+S_T)$. Then
\begin{equation*}
    f'(S)=-S_Tc_{yz}-c_y+Sc_{yy}-m\gamma c^{\gamma-1}c_y<0
\end{equation*}
by Lemma \ref{lem:c}. Clearly $f(0)=S_Tc_z+mc^\gamma>0$. Define
\begin{equation}
    \tilde{u}(y,z)\coloneqq u(c(y,z))=\begin{cases*}
        \frac{1}{1-\gamma}c(y,z)^{1-\gamma} & if $\gamma\neq 1$,\\
        \log(c(y,z)) & if $\gamma=1$.
    \end{cases*}\label{eq:vyz}
\end{equation}
Take any $\bar{y}>0$ and let $0<y<\bar{y}$. Using the chain rule and the monotonicity of $c$, we obtain
\begin{equation}
    \tilde{u}_y(y,z)=c(y,z)^{-\gamma}c_y(y,z)>c(\bar{y},z)^{-\gamma}c_y(y,z)\to \infty \label{eq:v_Inada}
\end{equation}
as $y\downarrow 0$ by Assumption \ref{asmp:c}. Using the definition of $f$, we obtain $f(S)c^{-\gamma}=S_T\tilde{u}_z-S\tilde{u}_y+m$.
Letting $S\uparrow e_{T-1}^y$ and using \eqref{eq:v_Inada}, we obtain $f(S)c^{-\gamma}\to -\infty$. Hence by the intermediate value theorem, there exists a unique $S_{T-1}\in (0,e_{T-1}^y)$ such that $f(S_{T-1})=0$. Therefore, there exists a unique equilibrium $\cS_{T-1}=\set{S_t}_{t=T-1}^\infty$ starting at $t=T-1$ that agrees with $\cS_T$ for $t\ge T$. The claim follows from backward induction.
\end{proof}

\subsection{Proof of Theorem \ref{thm:Gr}}

Take any equilibrium $\set{S_t}_{t=0}^\infty$. Using \eqref{eq:eq_r} and Assumption \ref{asmp:U}, the rent is
\begin{equation}
    r_t=m\frac{c^\gamma}{c_y}(e_1G^t-S_t,e_2G^{t+1}+S_{t+1}). \label{eq:r_explicit}
\end{equation}

We first show
\begin{equation}
    \limsup_{t\to\infty} r_t^{1/t}\le G^\gamma. \label{eq:rt_ub}
\end{equation}
Using the trivial bound $0\le S_t\le e_1G^t$, noting that $c$ is increasing in both arguments and $c_y$ is decreasing (increasing) in $y$ ($z$) by Lemma \ref{lem:c}, and using the homogeneity of $c$ and $c_y$, it follows from \eqref{eq:r_explicit} that
\begin{equation*}
    r_t\le m\frac{c(e_1G^t,(e_1+e_2)G^{t+1})^\gamma}{c_y(e_1G^t,e_2G^{t+1})}=me_1^\gamma\frac{c(1,G(1+w))^\gamma}{c_y(1,Gw)}G^{\gamma t}\eqqcolon \bar{r}G^{\gamma t}.
\end{equation*}
Taking the $1/t$-th power, we obtain $r_t^{1/t}\le G^\gamma \bar{r}^{1/t}$ for all $t$. Letting $t\to\infty$, we obtain \eqref{eq:rt_ub}.

We next show
\begin{equation}
    \liminf_{t\to\infty} G^{-t}S_t<e_1. \label{eq:St_liminf}
\end{equation}
Suppose to the contrary that $\liminf_{t\to\infty}G^{-t}S_t\ge e_1$. Using the trivial bound $S_t\le e_1G^t$, we obtain $\lim_{t\to\infty}G^{-t}S_t=e_1$. Take $\epsilon>0$ such that $G^{-t}S_t>e_1-\epsilon$ for large enough $t$. Then
\begin{equation*}
    \frac{r_t}{P_t}=\frac{r_t}{S_t-r_t}\le \frac{\bar{r}G^{\gamma t}}{(e_1-\epsilon)G^t-\bar{r}G^{\gamma t}}\sim \frac{\bar{r}}{e_1-\epsilon}G^{(\gamma-1)t}
\end{equation*}
as $t\to\infty$, so $\sum_{t=1}^\infty r_t/P_t<\infty$ because $\gamma<1$. By the Bubble Characterization Lemma \ref{lem:bubble}, there is a bubble. Using \eqref{eq:eq_R}, the homogeneity of $c$, and Assumption \ref{asmp:c}, the equilibrium interest rate satisfies
\begin{align*}
    R_t&=\frac{c_y}{c_z}(e_t^y-S_t,e_{t+1}^o+S_{t+1})\\
    &=\frac{c_y}{c_z}(e_1-G^{-t}S_t,G(e_2+G^{-t-1}S_{t+1}))\to \frac{c_y}{c_z}(0,G(e_1+e_2))=\infty
\end{align*}
as $t\to\infty$. Therefore, for any $R>G$, we can take $T>0$ such that $R_t\ge R>G$ for $t\ge T$. Letting $q_t>0$ be the Arrow-Debreu price, it follows that
\begin{equation*}
    q_tP_t=\left(q_T/\prod_{s=T}^{t-1}R_s\right)P_t\le q_T R^{T-t}e_1G^t=e_1q_TR^T(G/R)^t\to 0
\end{equation*}
as $t\to\infty$, so the no-bubble condition holds and there is no bubble, which is a contradiction.

Finally, we show
\begin{equation}
    \limsup_{t\to\infty}r_t^{1/t}\ge G^\gamma. \label{eq:rt_lb}
\end{equation}
Since \eqref{eq:St_liminf} holds, we can take $\bar{s}<1$ such that $S_t/e_t^y\le \bar{s}$ infinitely often. For such a subsequence, by a similar argument for proving \eqref{eq:rt_ub}, we obtain
\begin{equation*}
    r_t\ge m\frac{c((1-\bar{s})e_1G^t,e_2G^{t+1})^\gamma}{c_y((1-\bar{s})e_1G^t,(e_1+e_2)G^{t+1})}=me_1^\gamma\frac{c(1-\bar{s},Gw)^\gamma}{c_y(1-\bar{s},G(1+w))}G^{\gamma t} \eqqcolon \ubar{r}G^{\gamma t}.
\end{equation*}
Taking the $1/t$-th power, we obtain $r_t^{1/t}\ge G^\gamma \ubar{r}^{1/t}$. Letting $t\to\infty$, we obtain \eqref{eq:rt_ub}. The long-run rent growth rate \eqref{eq:Gr} follows from \eqref{eq:rt_ub} and \eqref{eq:rt_lb}. \hfill \qedsymbol

\subsection{Proof of Theorem \ref{thm:gamma<1f}}

\begin{proof}[Proof of Theorem \ref{thm:gamma<1f}\ref{item:f_w*f}]
By Lemma \ref{lem:c}, $(c_y/c_z)(y,G)$ is strictly decreasing in $y$ and has range $(0,\infty)$. Therefore, there exists a unique $y$ satisfying $(c_y/c_z)(y,G)=G^\gamma$. Since by Lemma \ref{lem:c} $c_y,c_z$ are homogeneous of degree 0, we have $(c_y/c_z)(1,G/y)=G^\gamma$, so $w_f^*=1/y$ uniquely satisfies \eqref{eq:w*f}.
\end{proof}

\begin{proof}[Proof of Theorem \ref{thm:gamma<1f}\ref{item:f_exist}] We divide the proof into several steps.

\setcounter{step}{0}

\begin{step}
    Derivation of an autonomous nonlinear difference equation.
\end{step}

By \eqref{eq:Pr<1f}, if a fundamental long-run equilibrium exists, then $S_t=P_t+r_t$ asymptotically grows at rate $G^\gamma$. Define the detrended variable $s_t\coloneqq S_t/(e_1^\gamma G^{\gamma t})$. Using the homogeneity of $c,c_y,c_z$, \eqref{eq:s_dynamics2} implies
\begin{equation}
    e_1^\gamma s_{t+1}G^{\gamma(t+1)}c_z-e_1^\gamma s_tG^{\gamma t}c_y+me_1^\gamma G^{\gamma t}c^\gamma, \label{eq:s_dynamics5}
\end{equation}
where $c,c_y,c_z$ are evaluated at
\begin{equation*}
    (y,z)=(1-s_te_1^{\gamma-1}G^{(\gamma-1)t},G(w+s_{t+1}e_1^{\gamma-1}G^{(\gamma-1)(t+1)})).
\end{equation*}
Dividing \eqref{eq:s_dynamics5} by $e_1^\gamma G^{\gamma t}$ and defining the auxiliary variable $\xi_t=(\xi_{1t},\xi_{2t})=(s_t,e_1^{\gamma-1}G^{(\gamma-1)t})$, it follows that \eqref{eq:s_dynamics2} can be rewritten as $\Phi(\xi_t,\xi_{t+1})=0$, where $\Phi:\R^4\to \R^2$ is given by
\begin{subequations}\label{eq:H2}
    \begin{align}
        \Phi_1(\xi,\eta)&=G^\gamma\eta_1c_z-\xi_1 c_y+mc^\gamma,\\
        \Phi_2(\xi,\eta)&=\eta_2-G^{\gamma-1}\xi_2
    \end{align}
\end{subequations}
and $c,c_y,c_z$ are evaluated at $(y,z)=(1-\xi_{1t}\xi_{2t},G(w+\xi_{1,t+1}\xi_{2,t+1}))$.\footnote{Obviously, \eqref{eq:H2} and \eqref{eq:H} are different because they correspond to the fundamental and bubbly equilibria, respectively.}

\begin{step}
    Existence and uniqueness of a fundamental steady state.
\end{step}

If a steady state $\xi_f^*$ of \eqref{eq:H2} exists, it must be $\xi_2=0$. Then the steady state condition is
\begin{equation*}
    G^\gamma sc_z-sc_y+mc^\gamma\iff s=m\frac{c^\gamma}{c_y-G^\gamma c_z},
\end{equation*}
where $c,c_y,c_z$ are evaluated at $(y,z)=(1,Gw)$. For $s>0$, it is necessary and sufficient that $c_y/c_z>G^\gamma$ at $(y,z)=(1,Gw)$. Since by Lemma \ref{lem:c} $c_y,c_z$ are homogeneous of degree 0 and $c_y/c_z$ is strictly increasing in $z$, there exists a fundamental steady state if and only if $w>w_f^*$.

\begin{step}
    Existence and local determinacy of equilibrium.
\end{step}

Define $\Phi$ by \eqref{eq:H2} and write $s=s^*$ to simplify notation. Noting that $\xi_f^*=(s^*,0)$, a straightforward calculation yields
\begin{align*}
    D_\xi \Phi(\xi_f^*,\xi_f^*)&=\begin{bmatrix}
        -c_y & -G^\gamma s^2c_{yz}+s^2c_{yy}-sm\gamma c^{\gamma-1}c_y\\
        0 & -G^{\gamma-1}
    \end{bmatrix},\\
    D_\eta \Phi(\xi_f^*,\xi_f^*)&=\begin{bmatrix}
        G^\gamma c_z & G^{\gamma+1}s^2c_{zz}-Gs^2c_{yz}+Gsm\gamma c^{\gamma-1}c_z\\
        0 & 1
    \end{bmatrix},
\end{align*}
where all functions are evaluated at $(y,z)=(1,Gw)$. Since $D_\eta \Phi$ is invertible, we may apply the implicit function theorem to solve $\Phi(\xi,\eta)=0$ around $(\xi,\eta)=(\xi_f^*,\xi_f^*)$ as $\eta=\phi(\xi)$, where
\begin{equation*}
    D\phi(\xi_f^*)=-[D_\eta \Phi]^{-1}D_\xi \Phi=\begin{bmatrix}
        \frac{c_y}{G^\gamma c_z} & \phi_{12}\\
        0 & G^{\gamma-1}
    \end{bmatrix}
\end{equation*}
and $\phi_{12}$ is unimportant. Since $c_y>G^\gamma c_z$, the eigenvalues of $D\phi$ are $\lambda_1=c_y/(G^\gamma c_z)>1$ and $\lambda_2=G^{\gamma-1}\in (0,1)$. Therefore, the steady state $\xi_f^*$ is a hyperbolic fixed point and the local stable manifold theorem \citep[Theorem 8.9]{TodaEME} implies that for any sufficiently large $e_1>0$ (so that $\xi_{20}=e_1^{\gamma-1}$ is close to the steady state value 0), there exists a unique orbit $\set{\xi_t}_{t=0}^\infty$ converging to the steady state $\xi_f^*$. However, by Assumption \ref{asmp:G}, choosing a large enough $e_1>0$ is equivalent to starting the economy at large enough $t=T$. Lemma \ref{lem:backward} then implies that there exists a unique equilibrium converging to the steady state regardless of the early endowments $\set{(e_t^y,e_t^o)}_{t=0}^{T-1}$.

\begin{step}
    The equilibrium objects have the order of magnitude in \eqref{eq:eqobj_gamma<1f} and the housing price equals its fundamental value.
\end{step}

The order of magnitude \eqref{eq:eqobj_gamma<1f} is obvious from $\lim_{t\to\infty}G^{-t}S_t=0$, the homogeneity of $c$, and Theorem \ref{thm:eq}. In equilibrium, both the housing price $P_t$ and rent $r_t$ asymptotically grow at rate $G^\gamma$. Therefore, $\sum_{t=1}^\infty r_t/P_t=\infty$, so there is no bubble by Lemma \ref{lem:bubble}.
\end{proof}

\begin{proof}[Proof of Theorem \ref{thm:gamma<1f}\ref{item:f_nonexist}]
Take any equilibrium. Because $h_t=1$ in equilibrium, in which case the utility $U(y,z,1)=u(c(y,z))+mu(1)$ is a monotonic transformation of $c(y,z)$, we can construct an equilibrium of an endowment economy without housing service in which agents have utility $c(y,z)$, the income of the young is $a_t\coloneqq e_t^y-r_t$, the income of the old is $b_t\coloneqq e_t^o$, and the asset pays dividend $r_t$. Condition \ref{item:necessity1} of Lemma \ref{lem:neccesity} follows from Assumptions \ref{asmp:U} and \ref{asmp:c}. Condition \ref{item:necessity2} of Lemma \ref{lem:neccesity} follows from Assumption \ref{asmp:G}, Theorem \ref{thm:Gr}, and $\gamma<1$. By \eqref{eq:Gr}, the long-run rent growth rate is $G_r\coloneqq G^\gamma$. Finally, since by Lemma \ref{lem:c} $c_y/c_z$ is strictly decreasing in $y$ (hence strictly increasing in $z$), if $w<w_f^*$, the autarky interest rate satisfies
\begin{equation*}
    R=\frac{c_y}{c_z}(e_1,e_2)=\frac{c_y}{c_z}(1,Gw)<\frac{c_y}{c_z}(1,Gw_f^*)=G^\gamma=G_r<G,
\end{equation*}
which is the bubble necessity condition \eqref{eq:necessity}. Therefore, all assumptions of Lemma \ref{lem:neccesity} are satisfied and the claim holds.
\end{proof}

\subsection{Proof of Theorem \ref{thm:gamma<1b}}

We divide the proof into several steps.

\setcounter{step}{0}

\begin{step}
    Existence and uniqueness of a bubbly steady state.
\end{step}

The proof of the existence and uniqueness of $w_b^*$ satisfying \eqref{eq:w*b} is identical to Theorem \ref{thm:gamma<1f}\ref{item:f_w*f}. Since $G>1$ and $\gamma<1$, it follows from \eqref{eq:w*f} and \eqref{eq:w*b} that
\begin{equation*}
    (c_y/c_z)(1,Gw_f^*)=G^\gamma<G=(c_y/c_z)(1,Gw_b^*).
\end{equation*}
Since $c_y/c_z$ is strictly increasing in $z$, we obtain $w_f^*<w_b^*$.

The steady state condition is $Gc_z-c_y=0$, where $c_y,c_z$ are evaluated at $(y,z)=(1-s,G(w+s))$. Using the homogeneity of $c_y,c_z$, this condition is equivalent to $(c_y/c_z)(y,G)=G$ for $y=\frac{1-s}{w+s}$, so the bubbly steady state is uniquely determined by
\begin{equation}
    \frac{1-s}{w+s}=\frac{1}{w_b^*}\iff s=\frac{w_b^*-w}{w_b^*+1}. \label{eq:s*<1}
\end{equation}
Since $s\in (0,1)$, a necessary and sufficient condition for the existence of a bubbly steady state is $w<w_b^*$.

\begin{step}
    Order of magnitude of equilibrium objects and asset pricing implications.
\end{step}

In any equilibrium converging to the bubbly steady state, by definition we have $S_t\sim se_1G^t$, where $s=s^*$ is the bubbly steady state. Therefore, \eqref{eq:yz<1b} follows from \eqref{eq:eq_yz}. Using \eqref{eq:eq_r} and Assumption \ref{asmp:U}, the rent is
\begin{equation}
    r_t=\frac{mu'(1)}{u'(c)c_y}=m\frac{c(e_t^y-S_t,e_{t+1}^o+s_{t+1})^\gamma}{c_y(e_t^y-S_t,e_{t+1}^o+s_{t+1})}.\label{eq:r_gamma}
\end{equation}
Substituting \eqref{eq:yz<1b} into \eqref{eq:r_gamma} and using the fact that $c$ is homogeneous of degree 1 and $c_y$ is homogeneous of degree 0, we obtain
\begin{equation*}
    r_t\sim me_1^\gamma \frac{c(1-s,G(w+s))^\gamma}{c_y(1-s,G(w+s))}G^{\gamma t}.
\end{equation*}
Since $r_t$ asymptotically grows at rate $G^\gamma<G$ because $\gamma<1$, we have $r_t/S_t\to 0$, so $P_t=S_t-r_t\sim S_t$ and \eqref{eq:Pr<1b} holds. Finally, \eqref{eq:R<1b} follows from \eqref{eq:R} and \eqref{eq:Pr<1b}.

Since the housing price $P_t$ and rent $r_t$ asymptotically grow at rates $G$ and $G^\gamma<G$, respectively, the rent-price ratio $r_t/P_t$ decays geometrically at rate $G^{\gamma-1}<1$. Therefore, $\sum_{t=1}^\infty r_t/P_t<\infty$, so there is a housing bubble by Lemma \ref{lem:bubble}.

\begin{step}
    Generic existence of equilibrium.
\end{step}

Define $\Phi$ by \eqref{eq:H} and write $s=s^*$ to simplify notation. Noting that $\xi_b^*=(s^*,0)$, a straightforward calculation yields
\begin{align*}
    D_\xi \Phi(\xi_b^*,\xi_b^*)&=\begin{bmatrix}
        -Gs c_{yz}-c_y+sc_{yy} & mc^\gamma \\
        0 & -G^{\gamma-1}
    \end{bmatrix},\\
    D_\eta \Phi(\xi_b^*,\xi_b^*)&=\begin{bmatrix}
        Gc_z+G^2s c_{zz}-Gsc_{yz} & 0\\
        0 & 1
    \end{bmatrix},
\end{align*}
where all functions are evaluated at $(y,z)=(1-s,G(w+s))$. If $D_\eta \Phi$ is invertible, we may apply the implicit function theorem to solve $\Phi(\xi,\eta)=0$ around $(\xi,\eta)=(\xi_b^*,\xi_b^*)$ as $\eta=\phi(\xi)$, where
\begin{equation*}
    D\phi(\xi_b^*)=-[D_\eta \Phi]^{-1}D_\xi \Phi=\begin{bmatrix}
        \phi_{11} & \phi_{12}\\
        0 & G^{\gamma-1}
    \end{bmatrix}
\end{equation*}
with
\begin{equation}
    \phi_{11}=\frac{Gs c_{yz}+c_y-sc_{yy}}{Gc_z+G^2s c_{zz}-Gsc_{yz}}\eqqcolon \frac{n}{d} \label{eq:h11}
\end{equation}
and $\phi_{12}$ is unimportant.  Therefore, $D\phi(\xi_b^*)$ has two real eigenvalues; one is $\lambda_1\coloneqq \phi_{11}$ and the other is $\lambda_2 \coloneqq G^{\gamma-1}\in (0,1)$ because $G>1$ and $\gamma\in (0,1)$.

Let us estimate $\lambda_1$. Using \eqref{eq:c_partial2}, the numerator of \eqref{eq:h11} is
\begin{align*}
    n&=c_y+s(Gc_{yz}-c_{yy})=c_y+s\left(-G\frac{y}{z^2}g''-\frac{1}{z}g''\right)\\
    &=c_y-s\frac{Gy+z}{z^2}g''=c_y-\frac{s(1+w)}{G(w+s)^2}g'',
\end{align*}
where we have used $(y,z)=(1-s,G(w+s))$. Similarly, the denominator is
\begin{align*}
    d&=Gc_z+Gs(Gc_{zz}-c_{yz})=Gc_z+Gs\left(G\frac{y^2}{z^3}g''+\frac{y}{z^2}g''\right)\\
    &=Gc_z+Gs\frac{y(Gy+z)}{z^3}g''=Gc_z+\frac{s(1-s)(1+w)}{G(w+s)^3}g''.
\end{align*}
At the steady state, we have $Gc_z=c_y=g'$, so
\begin{align}
    n&=g'-\frac{s(1+w)}{G(w+s)^2}g'', & d&=g'+\frac{s(1-s)(1+w)}{G(w+s)^3}g''. \label{eq:nd}
\end{align}
Since $s\in (0,1)$ and $g''<0$, clearly $n>d$.

We now study each case by the magnitude of the denominator $d$.

\begin{case}[$d>0$]
If $d>0$, then $0<d<n$ and hence $\lambda_1=n/d>1$. Since $\lambda_1>1>\lambda_2>0$, the steady state $\xi_b^*$ is a saddle point. The existence and uniqueness of an equilibrium path converging to the steady state $\xi_b^*$ follows by the same argument as in the proof of Theorem \ref{thm:gamma<1f}.
\end{case}
\begin{case}[$d=0$]
If $d=0$, the implicit function theorem is inapplicable and we cannot study the local dynamics by linearization.
\end{case}
\begin{case}[$d\in (-n,0)$]
If $-n<d<0$, then $\lambda_1=n/d<-1$. Therefore, $\xi_b^*$ is a saddle point and there exists a unique equilibrium by the same argument as in the case $d>0$.
\end{case}
\begin{case}[$d=-n$]
If $d=-n$, then $\lambda_1=n/d=-1$, the fixed point is not hyperbolic, and the local stable manifold theorem is inapplicable.
\end{case}
\begin{case}[$d<-n$]
If $d<-n$, then $\lambda_1=n/d\in (-1,0)$. Therefore, $\xi_b^*$ is a sink and there exist a continuum of equilibria by the same argument as in the case $d>0$.
\end{case}
In summary, there exists an equilibrium converging to the bubbly steady state except when $d=0$ or $d=-n$. Therefore, for generic $G$ and $w$, there exists an equilibrium. \hfill \qedsymbol

\subsection{Proof of Proposition \ref{prop:unique}}

We have already proved the uniqueness of the fundamental long-run equilibrium if $w>w_f^*$ in the proof of Theorem \ref{thm:gamma<1f}.

Suppose $w<w_b^*$. Let $s=\frac{w_b^*-w}{w_b^*+1}$ be the bubbly steady state and $(y,z)=(1-s,G(w+s))$. By the proof of Theorem \ref{thm:gamma<1b}, there exists a unique equilibrium converging to the bubbly steady state if $d\in (-n,0)\cup (0,\infty)$, where $d,n$ are the denominator and numerator in \eqref{eq:nd}. We rewrite this condition using the EIS defined by $\varepsilon=\frac{c_yc_z}{cc_{yz}}$. Using \eqref{eq:c_partial1}, \eqref{eq:c_partial2}, \eqref{eq:c_Euler}, and $Gc_z=c_y$ at the steady state, we obtain
\begin{equation*}
    \varepsilon=\frac{c_yc_z}{(yc_y+zc_z)c_{yz}}=\frac{c_y}{(Gy+z)c_{yz}}=-\frac{g'}{g''}\frac{G(w+s)^2}{(1-s)(1+w)}.
\end{equation*}
Therefore, \eqref{eq:nd} can be rewritten as
\begin{align}
    n&=\left(1+\frac{1}{\varepsilon}\frac{s}{1-s}\right)g',& d&=\left(1-\frac{1}{\varepsilon}\frac{s}{w+s}\right)g'. \label{eq:nd2}
\end{align}
Since $g'>0$, we have
\begin{align*}
    d=0&\iff \varepsilon=\frac{s}{w+s}=\frac{1-w/w_b^*}{1+w},\\
    n+d>0&\iff \varepsilon>\frac{s(1-w-2s)}{2(1-s)(w+s)}=\frac{1-w_b^*}{2}\frac{1-w/w_b^*}{1+w}.
\end{align*}
Therefore, the sufficient condition \eqref{eq:locdet_gamma<1} follows. \hfill \qedsymbol

\subsection{Proof of Proposition \ref{prop:efficient}}

To prove Proposition \ref{prop:efficient}, we need the following lemma.

\begin{lem}[Characterization of equilibrium efficiency]\label{lem:efficient}
Suppose Assumptions \ref{asmp:G}--\ref{asmp:c} hold and let $\set{S_t}_{t=0}^\infty$ be an equilibrium. Let $G_t=e_{t+1}^y/e_t^y$, $w_t=e_t^o/e_t^y$, and $s_t=S_t/e_t^y$. Let
\begin{equation}
    R_t=\frac{c_y}{c_z}(1-s_t,G_t(w_{t+1}+s_{t+1})) \label{eq:Rt}
\end{equation}
be the equilibrium risk-free rate and define the Arrow-Debreu price by $q_0=1$ and $q_t=1/\prod_{s=0}^{t-1}R_s$ for $t\ge 1$. Then the following statements are true.
\begin{enumerate}
    \item\label{item:efficient1} If $\liminf_{t\to\infty} R_t>G$, then the equilibrium is Pareto efficient.
    \item\label{item:efficient2} If $\limsup_{t\to\infty}s_t<1$, then the equilibrium is Pareto efficient if and only if
    \begin{equation}
        \sum_{t=0}^\infty \frac{1}{G^tq_t}=\infty. \label{eq:efficient_cond}
    \end{equation}
\end{enumerate}
\end{lem}

\begin{proof}[Proof of Lemma \ref{lem:efficient}]
Let $u_t(y,z)=c(y,G_tz)$ be the utility function in the detrended economy. Then the implied gross risk-free rate at the proposed allocation $(c_t^y,c_{t+1}^o)=(1-s_t,w_{t+1}+s_{t+1})$ is
\begin{equation*}
    \tilde{R}_t\coloneqq \frac{u_{ty}}{u_{tz}}(1-s_t,w_{t+1}+s_{t+1})=\frac{1}{G_t}\frac{c_y}{c_z}(1-s_t,w_{t+1}+s_{t+1})=\frac{R_t}{G_t}.
\end{equation*}
Therefore, the Arrow-Debreu price in the detrended economy is $\tilde{q}_t=\prod_{s=0}^{t-1}(G_s/R_s)$.

We now apply the results of \citet{BalaskoShell1980}. If $\liminf_{t\to\infty}R_t>G$, then by Assumption \ref{asmp:G} we can take $R>G$ such that $R_t\ge R>G=G_t$ for $t$ large enough. Then $G_t/R_t\le G/R<1$, so we have $\lim_{t\to\infty} \tilde{q}_t=0$. Proposition 5.3 of \citet{BalaskoShell1980} then implies that the equilibrium is efficient.

We next consider the case $\bar{s}\coloneqq \limsup_{t\to\infty}s_t<1$. We verify each assumption of Proposition 5.6 of \citet{BalaskoShell1980}. Since the partial derivatives of $c$ can be signed as in Lemma \ref{lem:c}, the Gaussian curvature of indifference curves are strictly positive. Since the time $t$ aggregate endowment of the detrended economy is $1+w_t$, which is bounded by Assumption \ref{asmp:G}, it follows that the Gaussian curvature of indifference curves within the feasible region (weakly preferred to endowments) is uniformly bounded and bounded away from 0 because $1-\bar{s}>0$. Therefore, assumptions (a) and (b) hold. Since $\bar{s}<1$ and $G_t$, $w_{t+1}$ are bounded, the gross risk-free rate \eqref{eq:Rt} can be uniformly bounded from above and away from 0. Therefore, assumption (c) holds. Assumption (d) holds because $w_t$ is bounded, and assumption (e) holds because $\liminf_{t\to\infty}(1-s_t)=1-\bar{s}>0$. Since all assumptions are verified, Proposition 5.6 of \citet{BalaskoShell1980} implies that the equilibrium is efficient if and only if
\begin{equation}
    \infty=\sum_{t=0}^\infty \frac{1}{\tilde{q}_t}=\sum_{t=0}^\infty \frac{1}{q_t}\prod_{s=0}^{t-1}(1/G_s).\label{eq:BS_cond}
\end{equation}
Since by Assumption \ref{asmp:G} we have $G_t=G$ for large enough $t$, \eqref{eq:BS_cond} is clearly equivalent to \eqref{eq:efficient_cond}.
\end{proof}

\begin{proof}[Proof of Proposition \ref{prop:efficient}]
Suppose $\gamma<1$ and consider any equilibrium. Using \eqref{eq:Rt}, Assumption \ref{asmp:G}, Lemma \ref{lem:c}, and $s_t\ge 0$, we obtain
\begin{equation}
    R_t=\frac{c_y}{c_z}(1-s_t,G_t(w_{t+1}+s_{t+1}))\ge \frac{c_y}{c_z}(1,Gw) \label{eq:Rt_lb}
\end{equation}
for large enough $t$. If $w\ge w_b^*$, then \eqref{eq:Rt_lb}, Lemma \ref{lem:c}, and \eqref{eq:w*b} imply
\begin{equation*}
    R_t\ge \frac{c_y}{c_z}(1,Gw)\ge \frac{c_y}{c_z}(1,Gw_b^*)=G.
\end{equation*}
Since $R_t\ge G$ eventually, the sequence $1/(G^tq_t)=\prod_{s=0}^{t-1}(R_s/G)$ is positive and bounded away from 0. Therefore, \eqref{eq:efficient_cond} holds, and the equilibrium is efficient.

Suppose $w<w_b^*$ and take any bubbly equilibrium converging to the bubbly steady state. By \eqref{eq:Pr<1b}, we can take $p>0$ such that $P_t\ge pG^t$ for large enough $t$. Then
\begin{equation*}
    G^tq_t=\frac{1}{p}q_tpG^t\le \frac{1}{p}q_tP_t\le \frac{1}{p}P_0
\end{equation*}
using \eqref{eq:P_iter}. Since $G^tq_t$ is positive and bounded above, $1/(G^tq_t)$ is positive and bounded away from 0, so \eqref{eq:efficient_cond} holds and the equilibrium is Pareto efficient.

Suppose $w<w_b^*$ and take the (unique) fundamental equilibrium. Then by Theorem \ref{thm:gamma<1f} we have $s_t\coloneqq S_t/(e_1G^t)\to 0$. Then \eqref{eq:Rt}, $s_t\to 0$, and $w<w_b^*$ imply that
\begin{equation*}
    \lim_{t\to\infty}R_t=\frac{c_y}{c_z}(1,Gw)<\frac{c_y}{c_z}(1,Gw_b^*)=G.
\end{equation*}
Therefore, we can take $R<G$ and $T>0$ such that $R_t\le R<G$ for $t\ge T$. Since
\begin{equation*}
    \frac{1}{G^tq_t}=\prod_{s=0}^{t-1}(R_s/G)\le \frac{1}{G^Tq_T}(R/G)^{t-T},
\end{equation*}
the sum $\sum_{t=0}^\infty 1/(G^tq_t)$ converges to a finite value, so by Lemma \ref{lem:efficient}\ref{item:efficient2} the equilibrium is inefficient.
\end{proof}

\subsection{Proof of Proposition \ref{prop:credit}}

By the discussion before the proposition, the available funds to the young at time $t$ is $\tilde{e}_t^y=e_t^y+\ell_t=(e_1+\ell)G^t$ and the available funds of the old at time $t$ is $\tilde{e}_t^o=e_t^o-G\ell_{t-1}=(e_2-\ell)G^t$ at interest rate $G$. Therefore, by Theorem \ref{thm:gamma<1b}, a bubbly long-run equilibrium exists if
\begin{equation*}
    0<\frac{e_2-\ell}{e_1+\ell}<w_b^*\iff w>\frac{\ell}{e_1}>\frac{w-w_b^*}{w_b^*+1},
\end{equation*}
which is \eqref{eq:lambda_cond}. Under this condition, because the old to young available funds ratio is $\tilde{w}\coloneqq \frac{e_2-\lambda e_1}{e_1+\lambda e_1}=\frac{w-\lambda}{1+\lambda}$, using \eqref{eq:Pr<1b} we obtain the asymptotic housing price
\begin{equation*}
    P_t\sim e_1(1+\lambda)\frac{w_b^*-\tilde{w}}{w_b^*+1}G^t=e_1\frac{(1+\lambda)w_b^*-(w-\lambda)}{w_b^*+1}G^t,
\end{equation*}
which simplifies to \eqref{eq:p_lambda}. \hfill \qedsymbol

\printbibliography

\newpage

\begin{center}
    {\Huge Online Appendix}
\end{center}

\section{Definition and characterization of bubbles}\label{sec:bubble}

\subsection{Definition of housing bubbles}

Following the standard definition of rational bubbles in the literature \citep{HiranoToda2024JME,HiranoToda2025EJW}, we define a housing bubble by a situation in which the housing price exceeds its fundamental value defined by the present value of rents. Let $R_t>0$ be the equilibrium gross risk-free rate. Let $q_t>0$ be the Arrow-Debreu price of date-$t$ consumption in units of date-0 consumption, so $q_0=1$ and $q_t=1/\prod_{s=0}^{t-1}R_s$. Since by definition $q_{t+1}=q_t/R_t$ holds, using \eqref{eq:R} we obtain the no-arbitrage condition
\begin{equation}
    q_tP_t=q_{t+1}(P_{t+1}+r_{t+1}). \label{eq:no-arbitrage}
\end{equation}
Iterating \eqref{eq:no-arbitrage} forward, for all $T>t$ we obtain
\begin{equation}
    q_tP_t=\sum_{s=t+1}^Tq_sr_s+q_TP_T. \label{eq:P_iter}
\end{equation}
Since $q_sr_s\ge 0$, letting $T\to \infty$ in \eqref{eq:P_iter}, we have $\sum_{s=t+1}^\infty q_sr_s\le q_tP_t$, so we may define the \emph{fundamental value} of housing by the present value of rents
\begin{equation*}
    V_t\coloneqq \frac{1}{q_t}\sum_{s=t+1}^\infty q_sr_s.
\end{equation*}
Letting $T\to\infty$ in \eqref{eq:P_iter}, we obtain the limit
\begin{equation}
    0\le \lim_{T\to\infty} q_TP_T=q_t(P_t-V_t). \label{eq:TVC}
\end{equation}
When the limit in \eqref{eq:TVC} equals 0, we say that the \emph{no-bubble condition} holds and the asset price $P_t$ equals its fundamental value $V_t$. When $\lim_{T\to\infty} q_TP_T>0$, we say that the no-bubble condition fails and the asset price contains a \emph{bubble}. Note that under rational expectations, we have either $P_t=V_t$ for all $t$ or $P_t>V_t$ for all $t$. Throughout the rest of the paper, we refer to an equilibrium with (without) a housing bubble a \emph{bubbly (fundamental) equilibrium}.

The economic meaning of $\lim_{T\to\infty}q_TP_T$ is that it captures a speculative aspect, that is, agents buy housing now for the purpose of resale in the future, in addition to receiving rents. The limit $\lim_{T\to\infty}q_TP_T$ captures its impact on current housing prices. When the no-bubble condition holds, the aspect of speculation becomes negligible and housing prices are determined only by factors that are backed in equilibrium, namely rents. On the other hand, when the no-bubble condition is violated, equilibrium housing prices contain a speculative aspect.

\subsection{Characterization and necessity of bubbles}

In general, proving the existence or nonexistence of bubbles is challenging because in the limit \eqref{eq:TVC}, both the Arrow-Debreu price $q_t$ and the housing price $P_t$ are endogenous. Here we discuss two useful results. Because the context does not matter, we consider a general asset that pays dividend $D_t\ge 0$ and trades at price $P_t$ (both in units of the consumption good). The first is the following Bubble Characterization Lemma due to \citet{Montrucchio2004}.

\begin{lem}[Bubble Characterization, \citealp{Montrucchio2004}]\label{lem:bubble}
If $P_t>0$ for all $t$, the asset price exhibits a bubble if and only if $\sum_{t=1}^\infty D_t/P_t<\infty$.
\end{lem}

\begin{proof}
See \citet[Lemma 2.1]{HiranoToda2025JPE}.
\end{proof}

Lemma \ref{lem:bubble} is useful because it does not involve the Arrow-Debreu price $q_t$ and provides a necessary and sufficient condition for the existence of bubbles.

The second result is the Bubble Necessity Theorem due to \citet{HiranoToda2025JPE}. To make the paper self-contained but to avoid technicalities, here we specialize the setting of \citet{HiranoToda2025JPE}. As in \S\ref{sec:model}, consider a two-period OLG model with a long-lived asset but assume that there is a single perishable good and the date-$t$ dividend $D_t\ge 0$ is exogenous (unlike our setting with endogenous rents). Define the long-run dividend growth rate by
\begin{equation}
    G_d\coloneqq \limsup_{t\to\infty} D_t^{1/t}. \label{eq:Gd}
\end{equation}
Let $(a_t,b_t)$ be the date-$t$ endowments of the young and old, and let $P_t\ge 0$ be the (endogenous) equilibrium asset price.

\begin{lem}[\citealp{HiranoToda2025JPE}, Theorem 2]\label{lem:neccesity}
Suppose that
\begin{enumerate*}
    \item\label{item:necessity1} the utility function $U(y,z)$ is continuously differentiable, homothetic, and quasi-concave, and
    \item\label{item:necessity2} the endowments satisfy $G^{-t}(a_t,b_t)\to (a,b)$ as $t\to\infty$, where $G>0$, $a>0$, and $b\ge 0$.
\end{enumerate*}
Define the long-run autarky interest rate by $R\coloneqq (U_y/U_z)(a,b)$. If
\begin{equation}
    R<G_d<G, \label{eq:necessity}
\end{equation}
then all equilibria are bubbly with asset price $P_t$ satisfying $\liminf_{t\to\infty}P_t/a_t>0$.
\end{lem}

Although the proof of Lemma \ref{lem:neccesity} is technical and we refer the reader to \citet{HiranoToda2025JPE}, the intuition is clear. If a fundamental equilibrium exists, the asset price must grow at the same rate as dividends, which is $G_d$. If $G_d<G$, the asset price becomes negligible relative to the size of the economy, and hence the allocation approaches autarky. With an autarky interest rate of $R<G_d$, the present value of dividends (and hence the asset price) becomes infinite, which is impossible. Therefore, a fundamental equilibrium cannot exist.

\section{Elasticity of substitution at most 1}\label{sec:gamma>=1}

The analysis in the main text focused on the empirically relevant case of $\gamma<1$ (Footnote \ref{fn:elasticity}), that is, the elasticity of substitution between consumption and housing $1/\gamma$ exceeds 1. For completeness, we present an analysis for the case $\gamma\ge 1$.

\subsection{Elasticity of substitution below 1}

We first consider the case $\gamma>1$, so the elasticity of substitution $1/\gamma$ is less than 1. In this case we cannot study the local dynamics around the steady state by linearization because the implicit function theorem is not applicable due to a singularity. Nevertheless, we may characterize the asymptotic behavior of all equilibria as follows.

\begin{prop}[Equilibrium with $\gamma>1$]\label{prop:gamma>1}
Suppose Assumptions \ref{asmp:G}--\ref{asmp:c} hold, $\gamma>1$, and let $w=e_2/e_1$. Then the following statements are true.
\begin{enumerate}
    \item In any equilibrium, the equilibrium objects satisfy
    \begin{subequations}\label{eq:eqobj>1}
    \begin{align}
        \lim_{t\to\infty}(c_t^y,c_t^o)/(e_1G^t)&=(0,1+w), \label{eq:yz>1}\\
        \lim_{t\to\infty}(P_t,r_t)/(e_1G^t)&=(0,1), \label{eq:Pr>1}\\
        \lim_{t\to\infty}R_t&=\infty. \label{eq:R>1}
    \end{align}
    \end{subequations}
    \item There is no housing bubble and the price-rent ratio converges to 0.
    \item Any equilibrium is Pareto efficient.
\end{enumerate}
\end{prop}

\begin{proof}
Let $\tilde{u}$ be defined by \eqref{eq:vyz}. Then the equilibrium dynamics \eqref{eq:s_dynamics3} can be written as
\begin{equation}
    Gs_{t+1}\tilde{u}_z=s_t\tilde{u}_y-me_1^{\gamma-1}G^{(\gamma-1)t}, \label{eq:s_dynamics6}
\end{equation}
where $\tilde{u}_y,\tilde{u}_z$ are evaluated at $(y,z)=(1-s_t,G(w+s_{t+1}))$. Define $\ubar{s}=\liminf_{t\to\infty}s_t$. Since $s_t\in (0,1)$, we have $0\le \bar{s}\le 1$. Take a subsequence of $(s_t,s_{t+1})$ such that $(s_t,s_{t+1})\to (\ubar{s},\tilde{s})$ for some $\tilde{s}$. Letting $t\to\infty$ in \eqref{eq:s_dynamics6} along this subsequence, we obtain
\begin{equation}
    0\le G\tilde{s}\tilde{u}_z(1-\ubar{s},G(w+\tilde{s}))=\ubar{s}\tilde{u}_y(1-\ubar{s},G(w+\tilde{s}))-\infty. \label{eq:s_dynamics_lim}
\end{equation}
Noting that $\tilde{u}_y(0,z)=\infty$ by \eqref{eq:v_Inada}, the only possibility for \eqref{eq:s_dynamics_lim} to hold is $\ubar{s}=1$. Then $s_t\to 1$, and
\begin{equation}
    \lim_{t\to\infty}\frac{S_t}{e_1G^t}=\lim_{t\to\infty}s_t=1. \label{eq:St_lim}
\end{equation}
Noting that $c_t^y=e_1G^t-S_t$ and $c_t^o=e_2G^t+S_t$, we obtain \eqref{eq:yz>1}. Using \eqref{eq:s_dynamics} and \eqref{eq:eq_r}, we obtain
\begin{equation}
    r_t=S_t-S_{t+1}\frac{U_z}{U_y}=S_t-S_{t+1}\frac{c_z}{c_y}, \label{eq:rt}
\end{equation}
where $c_y,c_z$ are evaluated at $(y,z)=(1-s_t,G(w+s_{t+1}))$. Dividing both sides of \eqref{eq:rt} by $e_1G^t$, letting $t\to\infty$, and using Lemma \ref{lem:c}, we obtain
\begin{equation*}
    \lim_{t\to\infty}\frac{r_t}{e_1G^t}=1-G\cdot 0=1.
\end{equation*}
Since $S_t=P_t+r_t$, we immediately obtain \eqref{eq:Pr>1}. Finally, the risk-free rate is
\begin{equation*}
    R_t=\frac{S_{t+1}}{P_t}=G\frac{S_{t+1}/(e_1G^{t+1})}{(S_t-r_t)/(e_1G^t)}\to G\frac{1}{1-1}=\infty,
\end{equation*}
which is \eqref{eq:R>1}.

Since $P_t\le S_t\sim e_1G^t$ grows at rate at most $G$ and the risk-free rate diverges to infinity (hence eventually exceeds the housing price growth rate), the no-bubble condition holds and there is no housing bubble. Using \eqref{eq:Pr>1}, we obtain $P_t/r_t\to 0$, so the price-rent ratio converges to 0. The Pareto efficiency of equilibrium follows from \eqref{eq:R>1} and Lemma \ref{lem:efficient}\ref{item:efficient1}.
\end{proof}

\subsection{Elasticity of substitution equal to 1} \label{subsec:gamma=1}

We next consider the case $\gamma=1$ (log utility), which is commonly used in applied theory. When $u(c)=\log c$, the difference equation \eqref{eq:s_dynamics3} reduces to
\begin{equation}
    Gs_{t+1}c_z=s_tc_y-mc, \label{eq:s_gamma=1}
\end{equation}
which is an autonomous nonlinear implicit difference equation. The following theorem shows that this difference equation admits a unique steady state, which defines a balanced growth path equilibrium.

\begin{prop}[Equilibrium with $\gamma=1$]\label{prop:gamma=1}
Suppose Assumptions \ref{asmp:G}--\ref{asmp:c} hold, $\gamma=1$, and let $w=e_2/e_1$. Then the following statements are true.
\begin{enumerate}
    \item\label{item:steady=1} There exists a unique steady state $s^*\in (0,1)$ of \eqref{eq:s_gamma=1}, which depends only on $G,w,c,m$.
    \item\label{item:order=1} There exists a unique balanced growth path equilibrium. The equilibrium objects satisfy
    \begin{subequations}\label{eq:eqobj_gamma=1}
        \begin{align}
            (c_t^y,c_t^o)&=((1-s^*)e_1G^t,(w+s^*)e_1G^t), \label{eq:yz=1}\\
            (P_t,r_t)&=\left(\frac{Gs^*c_z}{c_y}e_1G^t,m\frac{c}{c_y}e_1G^t\right), \label{eq:Pr=1}\\
            R_t&=\frac{c_y}{c_z}>G, \label{eq:R=1}
        \end{align}
    \end{subequations}
    where $c,c_y,c_z$ are evaluated at $(y,z)=(1-s^*,G(w+s^*))$.
    \item\label{item:price=1} In the equilibrium \eqref{eq:eqobj_gamma=1}, there is no housing bubble and the price-rent ratio $P_t/r_t$ is constant.
    \item Any equilibrium converging to the balanced growth path is Pareto efficient.
    \item\label{item:locdet=1} If in addition the elasticity of intertemporal substitution satisfies
    \begin{equation}
        \frac{1}{\varepsilon(y,z)}\coloneqq \frac{cc_{yz}}{c_yc_z}<\frac{1+w/s^*}{1+w}\left(1+Gw\frac{c_z}{c_y}\right) \label{eq:locdet_gamma=1}
    \end{equation}
    at $(y,z)=(1-s^*,G(w+s^*))$, then the equilibrium is locally determinate.
\end{enumerate}
\end{prop}

\begin{proof}
We divide the proof into several steps.

\setcounter{step}{0}
\begin{step}
Existence and uniqueness of $s^*$.
\end{step}

Letting $s_t=s_{t+1}=s$ in \eqref{eq:s_gamma=1} and rearranging terms, we obtain the steady state condition
\begin{equation}
    Gsc_z=sc_y-mc\iff \frac{Gc_z-c_y}{c}+\frac{m}{s}=0,\label{eq:s=1cond}
\end{equation}
where $c,c_y,c_z$ are evaluated at $(y,z)=(1-s,G(w+s))$. Define $f:(0,1)\to \R$ by
\begin{equation*}
    f(s)\coloneqq \log c(1-s,G(w+s))+m\log s.
\end{equation*}
Then \eqref{eq:s=1cond} is equivalent to $f'(s)=0$. Since $s\mapsto (1-s,G(w+s))$ is affine, the logarithmic function is increasing and strictly concave, and $m>0$, Proposition 11.4 of \citet[p.~150]{TodaEME} implies that $f$ is strictly concave. Clearly $f'(0)=\infty$. Letting $\tilde{u}(y,z)=\log c(y,z)$, an argument similar to the derivation of \eqref{eq:v_Inada} shows $\tilde{u}_y(0,z)=\infty$. Therefore, $f'(1)=-\infty$. Since $f$ is strictly concave, it has a unique global maximum $s^*\in (0,1)$, which satisfies $f'(s^*)=0$ and hence \eqref{eq:s=1cond}. Clearly this $s^*$ depends only on $G,w,c,m$.

\begin{step}
    Existence, uniqueness, and characterization of a balanced growth path.
\end{step}
In any balanced growth path equilibrium, we must have $S_t=s^*e_1G^t$ for some $s^*\in (0,1)$. The previous step establishes the existence and uniqueness of $s^*$. The consumption allocation \eqref{eq:yz=1} follows from \eqref{eq:eq_yz}, and Assumption \ref{asmp:G}. The rent in \eqref{eq:Pr=1} follows from \eqref{eq:eq_r}, Assumption \ref{asmp:U}, and Lemma \ref{lem:c}. Using \eqref{eq:s=1cond}, we obtain the housing price
\begin{equation*}
    P_t=S_t-r_t=e_1G^t\left(s-m\frac{c}{c_y}\right)=e_1G^t\frac{sc_y-mc}{c_y}=e_1G^t\frac{Gsc_z}{c_y},
\end{equation*}
which is \eqref{eq:Pr=1}. Using \eqref{eq:s=1cond}, we obtain the gross risk-free rate
\begin{equation*}
    R_t=\frac{S_{t+1}}{P_t}=\frac{se_1G^{t+1}}{e_1Gs(c_z/c_y)G^t}=\frac{c_y}{c_z}=G+\frac{mc}{sc_z}>G,
\end{equation*}
which is \eqref{eq:R=1}. Clearly the price-rent ratio is constant by \eqref{eq:Pr=1}. Since $R>G$, we obtain
\begin{equation*}
    \lim_{T\to\infty}R^{-T}P_T=\lim_{T\to\infty} e_1\frac{Gsc_z}{c_y}(G/R)^T=0,
\end{equation*}
so the no-bubble condition holds and there is no housing bubble. The Pareto efficiency of equilibrium follows from \eqref{eq:R=1} and Lemma \ref{lem:efficient}\ref{item:efficient1}.

\begin{step}
    Sufficient condition for local determinacy of equilibrium.
\end{step}

Define the function $\Phi:(0,1)\times (0,\infty)\to \R$ by
\begin{equation}
    \Phi(\xi,\eta)=G\eta c_z-\xi c_y+mc, \label{eq:Phi_gamma=1}
\end{equation}
where $c,c_y,c_z$ are evaluated at $(y,z)=(1-\xi,G(w+\eta))$. Then \eqref{eq:s_gamma=1} can be written as $\Phi(s_t,s_{t+1})=0$ and $\Phi(s,s)=0$ holds, where we write $s=s^*$. Assuming that the implicit function theorem is applicable and partially differentiating \eqref{eq:Phi_gamma=1}, we can solve the local dynamics as $s_{t+1}=\phi(s_t)$, where
\begin{align}
    \phi'(s)=-\frac{\Phi_\xi}{\Phi_\eta}&=-\frac{-Gsc_{yz}+sc_{yy}-(1+m)c_y}{-Gsc_{yz}+G^2sc_{zz}+G(1+m)c_z}\notag \\
    &=\frac{(1+m)c_y+Gsc_{yz}-sc_{yy}}{G(1+m)c_z-Gsc_{yz}+G^2sc_{zz}}\eqqcolon \frac{n}{d}. \label{eq:h'}
\end{align}
By exactly the same argument as in the proof of Theorem \ref{thm:gamma<1b}, we obtain
\begin{align*}
    n&=(1+m)c_y-\frac{s(1+w)}{G(w+s)^2}g'',\\
    d&=G(1+m)c_z+\frac{s(1-s)(1+w)}{G(w+s)^3}g''.
\end{align*}
If $\phi'(s)>1$, then $s=s^*$ is a source and hence the balanced growth path equilibrium is locally determinate.

We now seek to derive a sufficient condition for local determinacy. Since $g''<0$, we have
\begin{equation*}
    n-d>(1+m)(c_y-Gc_z)=m(1+m)\frac{c}{s}>0,
\end{equation*}
where we have used \eqref{eq:s=1cond}. Therefore, if $\Phi_\eta=d>0$, then $\phi'(s)=n/d>1$ and we have local determinacy.

Using \eqref{eq:h'}, \eqref{eq:c_partial2}, and $\sigma\coloneqq \frac{cc_{yz}}{c_yc_z}$, the sign of $\Phi_\eta$ becomes
\begin{align*}
    \sgn(\Phi_\eta)&=\sgn\left(-\frac{Gy+z}{z}sc_{yz}+(1+m)c_z\right) \\
    &=\sgn\left(-\frac{Gy+z}{z}s\sigma \frac{c_yc_z}{c}+(1+m)c_z\right) \\
    &=\sgn\left(-\frac{Gy+z}{z}s\sigma c_y+(1+m)c\right).
\end{align*}
Using \eqref{eq:c_Euler} and \eqref{eq:s=1cond}, we obtain
\begin{equation*}
    \sgn(\Phi_\eta)=\sgn\left(-\frac{Gy+z}{z}s\sigma c_y+yc_y+zc_z+sc_y-Gsc_z\right).
\end{equation*}
Substituting $(y,z)=(1-s,G(w+s))$, dividing by $c_y>0$, and rearranging terms, we obtain
\begin{equation*}
    \sgn(\Phi_\eta)=\sgn\left(-\frac{G(1+w)}{G(w+s)}s\sigma+1+Gw\frac{c_z}{c_y}\right).
\end{equation*}
Therefore, we have $\Phi_\eta>0$ if and only if
\begin{equation*}
    \frac{1}{\varepsilon}=\sigma<\frac{1+w/s}{1+w}\left(1+Gw\frac{c_z}{c_y}\right),
\end{equation*}
which is exactly \eqref{eq:locdet_gamma=1}.
\end{proof}

\section{Stylized facts}\label{sec:fact}

This appendix presents stylized facts regarding housing prices and rents.

We use the regional housing price data from \texttt{Realtor.com}, which provides detailed monthly data at the county level since July 2016.\footnote{\url{https://www.realtor.com/research/data/}} We use the median listing price in July because the sales volume tends to be higher in spring and summer.

Regional rents are the Fair Market Rents (FMRs) from the U.S. Department of Housing and Urban Development (HUD).\footnote{\url{https://www.huduser.gov/portal/datasets/fmr.html}} FMRs are defined by estimates of 40th percentile gross rents for standard quality units within a metropolitan area or non-metropolitan county and are available for housing units with 0--4 bedrooms. We use the values for three bedrooms.

The number of housing units is ``All housing units'' in Quarterly Estimates of the Total Housing Inventory for the United States from the Census Bureau,\footnote{\url{https://www.census.gov/housing/hvs/data/histtab8.xlsx}} which is available since 1965.

Figure \ref{fig:units_GDP} shows the time series of U.S. real GDP and the total number of housing units, where we normalize the values in 1965 to 1. We can see that GDP growth is faster, justifying our assumption $G>1$ in the model.

\begin{figure}[htb!]
    \centering
    \includegraphics[width=0.6\linewidth]{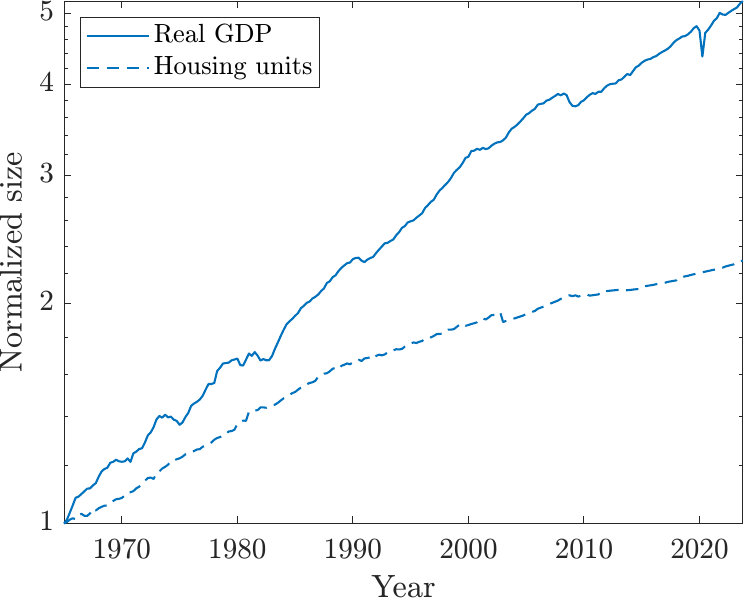}
    \caption{Growth of GDP and housing units.}
    \label{fig:units_GDP}
\end{figure}

Let $r_{it}$ and $P_{it}$ be the rent and housing price in county $i$ in year $t$ constructed above. Figure \ref{fig:rent_price} plots $\log P_{it}$ against $\log r_{it}$ for the year 2023 and estimates
\begin{equation*}
    \log P_{it}=\alpha+\beta \log r_{it}+\epsilon_{it},\quad i=1,\dots,I
\end{equation*}
by ordinary least squares (OLS) regression. The results for other years are all similar. Although this picture only documents correlation, the coefficient $\hat{\beta}=1.46>1$ corresponds to $1/\gamma$ in the model.

\begin{figure}[htb!]
\centering
    \includegraphics[width=0.6\linewidth]{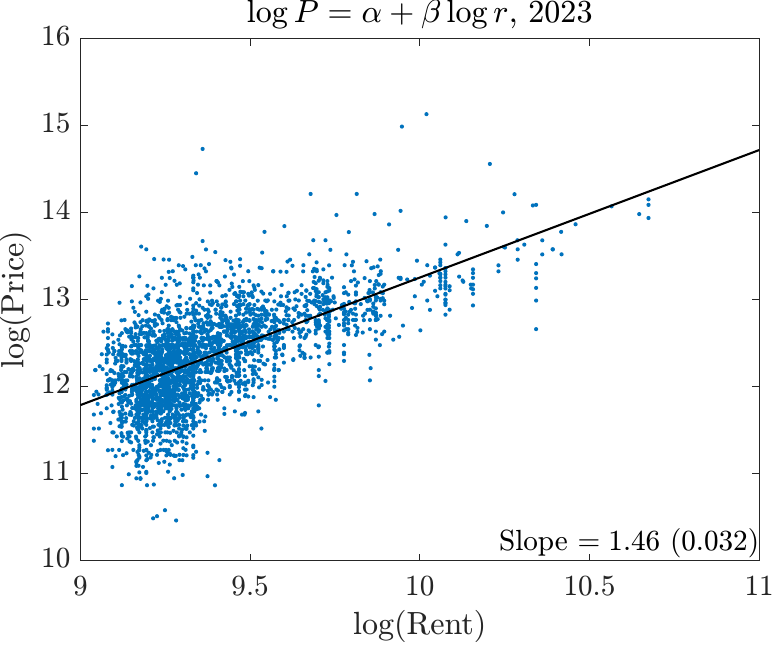}
    \caption{Rent and housing price across counties.}\label{fig:rent_price}
\end{figure}

\end{document}